\documentclass[a4paper,UKenglish,cleveref, autoref, thm-restate]{lipics-v2021}

\usepackage{xspace}
\usepackage[svgnames]{xcolor} 
\title{Universal Solvability for Robot Motion Planning on Graphs} 

\titlerunning{Universal Solvability for Robot Motion Planning on Graphs} 

\author{Anubhav Dhar}{Indian Institute of Technology Kharagpur, India}{anubhavldhar@gmail.com}{https://orcid.org/0009-0006-5922-8300}{}
\author{Ashlesha Hota}{Indian Institute of Technology Kharagpur, India}{ashleshahota@gmail.com}{https://orcid.org/0009-0009-8805-4583}{}
\author{Sudeshna Kolay}{Indian Institute of Technology Kharagpur, India}{skolay@cse.iitkgp.ac.in}{https://orcid.org/0000-0002-2975-4856}{}
\author{Pranav Nyati}{Indian Institute of Technology Kharagpur, India}{pranavnyati26@gmail.com}{https://orcid.org/0009-0007-6792-2420}{}
\author{Tanishq Prasad}{Indian Institute of Technology Kharagpur, India}{tanishqprasad1003@gmail.com}{https://orcid.org/0009-0007-5893-1317}{}

\authorrunning{A. Dhar, A. Hota, S. Kolay, P. Nyati, T. Prasad} 

\Copyright{} 



\category{} 

\relatedversion{}

\nolinenumbers 

\newcommand{\defproblem}[3]{
  \vspace{1mm}
\noindent\fbox{
  \begin{minipage}{0.96\textwidth}
  \begin{tabular*}{\textwidth}{@{\extracolsep{\fill}}lr} #1 \\ \end{tabular*}
  {\bf{Input:}} #2  \\
  {\bf{Question:}} #3
  \end{minipage}
  }
  \vspace{1mm}
}

\newcommand{\amap}[1]{{\mathcal{T}_{#1(R)}}}
\newcommand{\amapset}[1]{{\mathcal{T}_{#1}}}

\usepackage{todonotes}
\usepackage{algorithm}
\usepackage{algpseudocode}
\usepackage{mathtools}

\newcommand{\SOLR}{\textsc{Universal Solvability of Robot Motion Planning on Graphs (USolR)}}

\newcommand{\OO}{\mathcal{O}}

\newcommand{\USOLR}{\textsc{USolR}\xspace}
\newcommand{\FRMP}{\textsc{FRMP}\xspace}
\newcommand{\Ex}{\mathcal{E}}

 \definecolor{babyblue}{rgb}{0.54, 0.81, 0.94}
 \definecolor{b1}{rgb}{0.63, 0.79, 0.95}
 \definecolor{b2}{rgb}{0.74, 0.83, 0.9}
 \definecolor{b3}{rgb}{0.67, 0.9, 0.93}
 \definecolor{gentlegreen}{rgb}{0.00, 0.51, 0.00}

\begin{document}

\maketitle

\begin{abstract}
We study the $\SOLR$ problem: given an undirected graph $G = (V, E)$ and $p \leq |V|$ robots, determine whether any arbitrary configuration of the robots can be transformed into any other arbitrary configuration via a sequence of valid, collision-free moves.

We design a canonical accumulation procedure that maps arbitrary configurations to configurations that occupy a fixed subset of vertices, enabling us to analyse configuration reachability in terms of equivalence classes. We prove that in instances that are not universally solvable, at least half of all configurations are unreachable from a given one, and leverage this to design an efficient randomized algorithm with one-sided error, which can be derandomized with a blow-up in the running time by a factor of $p$. Further, we optimize our deterministic algorithm by using the structure of the input graph $G(V, E)$, achieving a running time of $\mathcal{O}(p\cdot (|V| + |E|))$ in sparse graphs and $\mathcal{O}(|V| + |E|)$ in dense graphs.

Finally, we consider the \textsc{Graph Edge Augmentation for Universal Solvability (EAUS)} problem where given a connected graph $G$ that is not universally solvable for $p$ robots, the question is to check if for a given budget $\beta$, at most $\beta$ edges can be added to $G$ to make it universally solvable for $p$ robots. We provide an upper bound of $p-2$ on $\beta$ for general graphs. On the other hand, we also provide examples of graphs that require  $\Theta (p)$ edges to be added in order to achieve universal solvability for $p$ robots. Next, we consider the augmentation problem \textsc{Graph Vertex and Edge Augmentation for Universal Solvability (VEAUS)} involving the addition of $\alpha$ vertices along with the addition of $\beta$ many edges to make a graph universally solvable for $p$ robots, and we provide further lower bounds for $\alpha$ and $\beta$.
\end{abstract}
\newpage
\section{Introduction}\label{sec:intro}
Coordinating the motion of multiple autonomous robots is one of the fundamental problems in robotics, artificial intelligence, and computational geometry. The broad aim is to coordinate multiple robots to navigate their environment without collisions. In this study, we model the environment, or workspace, as an undirected graph \( G(V, E) \), where the set \( V \) denotes the set of $n$ vertices representing discrete locations that the robots can occupy, and the edges \( E \) represent $m$ paths connecting two distinct locations. We consider \( p \leq |V| \) robots, where each robot occupies exactly one vertex, and no two robots occupy the same vertex at any time. At each step, a robot is allowed to either stay at its current vertex or move only to an adjacent vertex.
Additionally, two robots occupying adjacent vertices cannot swap positions along the same edge in one move.
A configuration of the system is defined as a snapshot at time \( t \), specifying the positions of all robots on the graph at that time. A move is defined as a transition between two such configurations. Let \( S \) denote the initial configuration, and let \( T \) denote the final desired configuration. We require the robots to move from \( S \) to \( T \) without any collisions or deadlocks.

To ground this in a practical setting, consider a warehouse where autonomous mobile robots move items between storage shelves and packing stations. The warehouse can be modeled as a graph \( G(V, E) \), where each vertex represents a location (such as storage shelves and packing stations), and each edge denotes a traversable path between two locations (such as lanes where the robots move). Suppose there are \( p \) robots operating simultaneously, each occupying a unique location. Over time, the system may need to rearrange the robots from one configuration \( S \) (e.g., initial positions) to any other configuration \( T \) (e.g., task-specific positions). To perform the desired task or operation efficiently, the system requires that such transformations be possible for any pair of configurations. This motivates us to solve the $\SOLR$ problem. 

We study the problem $\USOLR$ in which we address the following decision question: Given a graph $G(V, E)$ and an integer $p$, is the system capable of transforming any configuration $S$ to any other configuration $T$ of $p$ robots such that there exists a sequence of valid moves from $S$ to $T$? We denote graphs satisfying such a property as \emph{universally solvable graphs}. However, not all graphs are universally solvable. There are graph structures, such as cut edges, 
that prevent certain configurations from being reachable from other configurations. In such 
cases, the graph $G$ is not universally solvable, 
as there exist configurations $S$ and $T$ between which no valid sequence of moves exists.

This raises a natural follow-up question: Can we augment a graph $G$ which is not universally solvable for $p$ robots by adding a small number of edges so that the resulting graph becomes universally solvable? More precisely, given a budget $\beta$, we seek to determine if it is possible to make $G$ universally solvable for $p$ robots by adding at most $\beta$ edges.
It is also interesting to study the following question: Consider a graph that is not universally solvable and cannot achieve it through the addition of a small number of edges. Is universal solvability for such a graph possible if we also allow addition of $\alpha$ vertices, along with addition of $\beta$ many edges. For example, for simple graphs like star graphs with a robot placed on each vertex, it is possible to show that while the graph is neither universal solvable nor does it become universally solvable under a small budget of edge additions, it becomes universally solvable upon addition of just two vertices as neighbours only to the central vertex.
Designing such minimal augmentations that achieve full reconfigurability but also have minimal structural changes are an important consideration in practical settings like warehouses, as physical modifications may be costly.

The motivation for our work also stems from prior studies that investigate coordinated motion planning in restricted environments, such as trees~\cite{auletta1999linear,masehian2009solvability} and bi-connected graphs~\cite{kornhauser1984coordinating,wilson1974graph}. These works provide efficient algorithms for checking the feasibility of configuration transitions on a given pair of $S$ and $T$; however, they typically do not address whether \emph{all} configurations are mutually reachable. While group-theoretic approaches have been explored in~\cite{kornhauser1984coordinating} to examine graph classes for which any pair of configurations over the same vertex set are mutually reachable, the notion of \textit{universal solvability}, a property that guarantees the system can adapt to any dynamic task assignment (i.e., can reconfigure from any configuration to any other), remains underexplored, especially in general graphs. In this paper, we aim to bridge this gap by analysing the $\USOLR$ problem from both a structural and algorithmic perspective for general graphs. 

\subsection{Related Work}
Robot motion planning and path finding have received significant attention, both in theoretical and applied contexts over the past few decades~\cite{ismail2018survey,jahanshahi2018robot,qin2023review}. Closely related problems are: token swapping~\cite{bonnet2018complexity,yamanaka2015swapping} and pebble motion on graphs~\cite{fujita2015pebble,goraly2010multi, kato2021pebble,kornhauser1984coordinating,surynek2009application}. In the token swapping problem, the tokens are swapped on the endpoints of an edge whereas in the pebble motion problem, the pebbles move to a neighbouring unoccupied vertex.
More recently, coordinated motion planning has been studied with respect to minimising various target objectives, such as the makespan of the schedule of the robots~\cite{7949109,eiben_et_al:LIPIcs.SoCG.2023.28,7342901} and the total distance travelled by all robots~\cite{deligkas_et_al:LIPIcs.ICALP.2024.53, eiben_et_al:LIPIcs.SoCG.2023.28,10.5555/3535850.3535905,7342901}. 
Problems related to robot motion planning arise in robotics~\cite{chung2018survey}, motion planing~\cite{demaine2019coordinated}, puzzle design~\cite{loyd1959mathematical}, and distributed systems~\cite{1612674}, and have motivated extensive research on their complexity, solvability, and algorithmic strategies. Now, we discuss a few key contributions related to the feasibility variants of robot motion planning.

Kornhauser, Miller, and Spirakis~\cite{kornhauser1984coordinating} study the {\sc Pebble Motion} problem on graphs, where \( p < n \) labelled pebbles are placed on distinct vertices of an $n$-vertex graph \( G \), and the goal is to reach a target configuration via moves to adjacent unoccupied vertices. They reduce the problem to a \emph{permutation group} setting, where valid move sequences form a group \( R(P) \). For various graph classes (e.g., biconnected, separable, trees), \( R(P) \) is shown to be transitive or decomposable into transitive components. The authors prove that \( R(P) \) contains either the \emph{alternating group} \( \mathbf A_p \) or the \emph{symmetric group} \( \mathbf S_p \), depending on the structure of \( G \), and show that when a generator of prime length is available, the diameter of the group is sub-exponential, linking algebraic properties with motion planning. Auletta, Monti, Parente, and Persiano~\cite{auletta1999linear} address the {\sc Pebble Motion} problem on trees and present a linear-time algorithm that reduces it to the {\sc Pebble Permutation} problem, where the set of occupied vertices remains the same and the goal is to realise a given permutation of the pebbles in these vertices through a sequence of valid moves. Their approach yields an \( \OO(n) \)-time decision algorithm and an \( \OO(n + \ell) \)-time constructive algorithm, where \( \ell \) is the length of the move sequence and $n$ is the number of vertices.

Yu and Rus~\cite{yu2014pebblemotiongraphsrotations} generalise the classical {\sc Pebble Motion} problem to \textit{pebble motion with rotations} ({\sc PMR}) for $p$ pebbles on an $n$-vertex graph, where in addition to simple moves into unoccupied vertices (\( p < n \)), synchronous cyclic rotations along disjoint cycles are also allowed. In the fully packed case (\( p = n \)), only rotations are feasible. They model reachable configurations as elements of a subgroup \( \mathbf G \le \mathbf S_n \), generated by these rotations. For 2-edge connected but not 2-connected graphs, they prove that \( \mathbf G \) contains either the alternating group \( \mathbf A_n \) or the symmetric group \( \mathbf S_n \), and establish an upper bound of \( \text{diam}(\mathbf G) = \OO(n^2) \), ensuring reachability in polynomial time. These results extend to general 2-edge-connected graphs and enable a unified framework where feasibility can be tested in linear time, and a complete plan can be computed in \( \OO(n^3) \) time for all values of \( p \leq n \).

Goraly and Hassin~\cite{goraly2010multi} study the {\sc Multi-colour Pebble Motion} problem, where pebbles of different colours must be rearranged to match a target colour configuration. They extend earlier work by analysing the coloured setting and provide a linear-time feasibility test. Their method reduces the problem to a tree with transshipment vertices (vertices that facilitate movement, however, cannot hold pebbles). Adapting the algorithm by Auletta et al.~\cite{auletta1999linear}, they use equivalence classes of vertices to efficiently handle colour constraints and movement rules, resolving an open problem from earlier literature.

Masehian and Nejad~\cite{masehian2009solvability} investigate {\sc Multi-Robot Motion Planning} on trees, aiming to determine whether any initial configuration can reach any target configuration without collisions. They introduce the notion of solvable trees, partially solvable trees, and minimal solvable trees, and analyse conditions under which global and local interchanges are feasible. Their approach uses the concept of influence zones and proposes a linear-time algorithm based on computing each robot's maximum reachability space.

While most existing works focus on checking the feasibility of transforming a specific start configuration \( S \) into a target configuration \( T \), relatively little attention has been paid to the question of \textit{universal solvability}: whether all configurations on a given graph are mutually reachable. This motivates our study, which aims to analyse the universal solvability of general graphs and bridge this gap in the literature.

\subsection{Our Contributions}

In this work, we address the $\USOLR$ problem: given an undirected graph $G(V, E)$ and an integer $p \leq |V|$, we ask whether all configurations of $p$ robots placed on distinct vertices can be transformed into each other using valid moves, without collisions or swaps. We model configurations as injective maps from robots to vertices. Valid moves are defined as synchronous motions along either simple paths or simple cycles. 

The group theoretic results in the works of Yu and Rus~\cite{yu2014pebblemotiongraphsrotations} motivate us to look into equivalence classes of configurations that are reachable from each other. 
We prove that the sizes of all such equivalence classes are equal. This results in the fact that the number of reachable configurations is a factor of the number of all possible configurations when $p=|V|$.
For $p \le |V|$, we use a similar argument to arrive at the same conclusion. 
This result is used to design a simple and efficient randomized algorithm solving $\USOLR$. 

To analyse this problem structurally, we introduce a deterministic accumulation procedure in~\cref{sec:accumulation} that transforms any configuration to a canonical one where robots occupy a fixed vertex set $V_p \subseteq V$, $|V_p| = p$. This accumulation procedure is not used as a subroutine in the algorithm design but used for analyzing the algorithms later on.
Recall that our definition of the equivalence relation over configurations is based on mutual reachability and we show that each equivalence class has the same size. Using this, we prove that in any instance that is not universally solvable, at least half of the configurations are unreachable from a fixed configuration. 

With this insight, we design a randomized algorithm that samples a configuration uniformly at random and uses a linear-time subroutine to test reachability from the identity configuration; this algorithm has a one-sided error and runs in $\mathcal{O}(|V| + |E|)$ time (\cref{sec:randomized}). While the randomized algorithm is not directly used in the final deterministic solution, it is independently interesting due to its linear-time performance on general graphs. This is in sharp contrast to the best known deterministic algorithm for universal solvability, which requires quadratic time even for planar graphs \cite{yu2014pebblemotiongraphsrotations}.

In~\cref{sec:deterministic-algo}, we derandomize the previous algorithm that reduces the number of required reachability checks from $p!$ (i.e., a trivial derandomization of the previous algorithm) to just $p - 1$ by observing that reachability of adjacent transpositions (swaps between robots $t$ and $t+1$ for some $t$) suffices to check universal solvability. 

We also propose an optimized algorithm that takes into account structural graph properties (such as connectivity and the existence of sufficiently large 2-connected components). This algorithm runs in $\mathcal{O}(p \cdot (|V| + |E|))$ time in the worst case and improves to $\mathcal{O}(|E|)$ when the graph is sufficiently dense with $|E| = \omega(p \cdot |V|)$ (\cref{sec:optimised_deterministic}). 

Further, we consider the \textsc{Graph Edge Augmentation for Universal Solvability (EAUS)} problem where given a connected graph $G$ that is not universally solvable for $p$ robots, the question is to check if for a given budget $\beta$ at most $\beta$ edges can be added to $G$ to make it universally solvable for $p$ robots (\cref{sec:graph-augment}). For a graph that is 2-edge-connected but not 2-connected, it may not necessarily be universally solvable for $p=|V|$ robots. However, we show that the addition of a single edge ensures universal solvability in such case. In case of simple cycles, we show that the addition of a chord suffices to render the given instance universally solvable for all $p\leq |V|$. In contrast, for 1-edge-connected graphs, the augmentation problem becomes more challenging. We prove that $\Theta (p)$ edges are required in the worst case for star graphs. On the positive side, we show that $\beta \leq p-2$ for all graphs.

Finally, we consider the problem {\sc Graph Vertex and Edge Augmentation for Universal Solvability (VEAUS)} where along with $\beta$ edge additions, we allow the addition of an $\alpha$ vertices to the input graph and ask if it is sufficient to make the graph universally solvable for $p$ robots (also in~\cref{sec:graph-augment}). We provide lower bounds for $\alpha$ and $\beta$ for this problem. An interesting consequence of these lower bounds is that there is an infinite class of graphs for which a budget of $\alpha = \Omega(\sqrt{|V|})$ is required, even when $\beta = \Theta(\sqrt{|V|})$, for universal solvability of a graph $G(V,E)$ with $p = |V|$ robots.


\section{Preliminaries}\label{sec:prelims}

\subparagraph*{Basic Notations.} We denote the set of positive integers by $\mathbb{N}$ and the set of integers by $\mathbb{Z}$. For $k \in \mathbb{N}$, we denote the set $\{1, 2, \ldots, k\}$ by $[k]$. For a set $S$, we denote its power set (family of all subsets of S) by $2^S = \{S' \mid S' \subseteq S\}$. For $k \in \mathbb N$, we denote the family of subsets of $S$ of size $k$ by $\binom{S}{k}$. For a function $f:X \to Y$, and a subset $X' \subseteq X$, the set of images of $X'$ is denoted by $f(X') = \{f(x) \mid x \in X'\}$. Further, let $\mathfrak{S}_p$ be the set of all permutations from $[p]$ to $[p]$ ($| \mathfrak{S}_p| = p!)$. For functions $f: X_1 \to X_2$, and $g: X_2 \to X_3$, we denote their composition as $f \circ g : X_1 \to X_3$, i.e., $f \circ g (x) = f(g(x))$ for all $x \in X_1$.

\subparagraph*{Graph Theory.} We denote by $G(V,E)$, an undirected graph with vertex set $V$ and edge set $E$. Throughout this paper, we only consider undirected graphs. A cut vertex (resp. edge) of a graph is a vertex (resp. edge) whose deletion increases the number of connected components. The edge connectivity, $\lambda(G)$ of a graph $G(V, E)$ is the minimum number of edges required to be deleted to make the graph disconnected. We say that a graph $G(V, E)$ is $k$-edge connected if $\lambda(G) \ge k$. The vertex connectivity, $\kappa(G)$ of a graph $G(V, E)$, is the minimum number of vertices required to be deleted to make the graph disconnected. We say that a graph $G(V, E)$ is $k$- connected if $\kappa(G) \ge k$. 

\subparagraph*{Motion Planning.} We consider an underlying undirected graph $G(V,E)$ with $|V| = n$ and $|E| = m$. Consider a set $R$ of $p \in \mathbb N$ robots. For convenience, we assume the robots are named $1$, $2$, $\ldots$, $p$; hence $R = [p]$. We use both $R$ and $[p]$ to denote the set of robots, depending on the context. A \emph{configuration}, $S: R \to V$ is an injective map from the set of robots to the set of vertices of the underlying graph $G(V, E)$; for robot $i$, $S(i)$ denotes the vertex occupied by it. The configuration is an injective map as it models the constraint that no two robots can occupy the same vertex in a configuration, i.e., $i \neq j \implies S(i) \neq S(j)$. We denote by $\mathcal S_G = \{S : R \to V \mid S \text{ is a configuration}\}$ the set of all configurations of the robots in $R$ that are possible on the graph $G$. Note that $|\mathcal{S}_G| = \frac{n!}{(n-p)!}$.

\subparagraph*{Valid Moves.} We now formally define the notion of valid/feasible moves in terms of pairs of configurations that can be reached from each other by one such move. We introduce three different types of valid moves as follows:
\begin{itemize}
    \item \textbf{Simple Path Move:} A simple path move is defined as a pair of configurations that can be achieved from one another by sliding robots along a path synchronously, as shown in \cref{fig:simple_path}. Formally, for configurations $S, S'$, the pair $(S,S')$ defines a simple path move if there is a path $P = (u_0, u_1, u_2, \ldots, u_k)$ in $G$ such that, 
    \begin{itemize}
        \item $u_k \notin S(R)$ (empty in $S$) and $u_0 \notin S'(R)$ (empty in $S'$).
        \item For all $i \in R$ such that $S(i) \notin P$, we have $S'(i) = S(i)$ (robots outside the path $P$ do not move).
        \item For all $i \in R$ such that $S(i) = u_j$, with $j \in \{0, 1, \ldots, k-1\}$, we have $S'(i) = u_{j+1}$ (robots in $P$ synchronously move one unit each).
    \end{itemize}
    We say that $S'$ is obtained by pushing $S$ along the path $P$. 

    \begin{figure} [ht]
        \centering
        \includegraphics[scale=0.8]{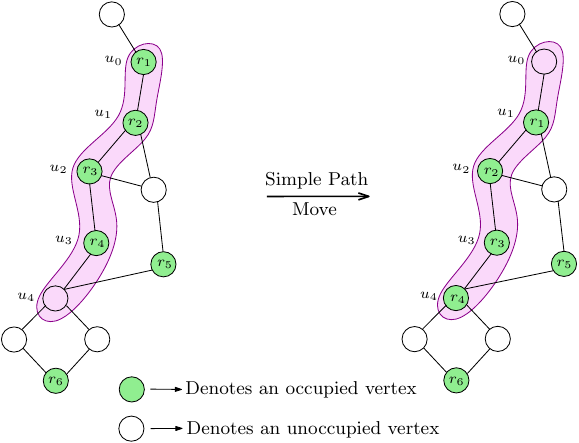}
        \caption{Simple path move, where $r_i$ denotes robot $i$.}\label{fig:simple_path}
    \end{figure}
    
    \item \textbf{Simple Rotation Move:}
    A simple rotation move is a move along a simple cycle (all whose vertices are occupied by robots) such that all the vertices in that cycle synchronously move to their immediate neighbour in either clockwise or anticlockwise direction in one such move, as depicted in \cref{fig:simple_rot}. Formally, for configurations $S, S'$, the pair $(S,S')$ defines a simple rotation if there is a cycle $C = (u_0, u_1, u_2, \ldots, u_k = u_0)$ in $G$ such that, 
    \begin{itemize}
        \item $\{u_0, \ldots, u_{k-1}\} \subseteq S(R)$ (occupied vertices in $S$) and $\{u_0, \ldots, u_{k-1}\} \subseteq S'(R)$ (occupied vertices in $S'$).
        \item For all $i \in R$ such that $S(i) \notin C$, we have $S'(i) = S(i)$ (robots outside the cycle $C$ do not move).
        \item For all $i \in R$ such that $S(i) = u_j$, with $j \in \{0, 1, \ldots, k-1\}$, we have $S'(i) = u_{j+1}$ (robots in $C$ synchronously move one unit each).
    \end{itemize} 

    \begin{figure} [ht] 
    \centering 
    \includegraphics[scale=0.81]{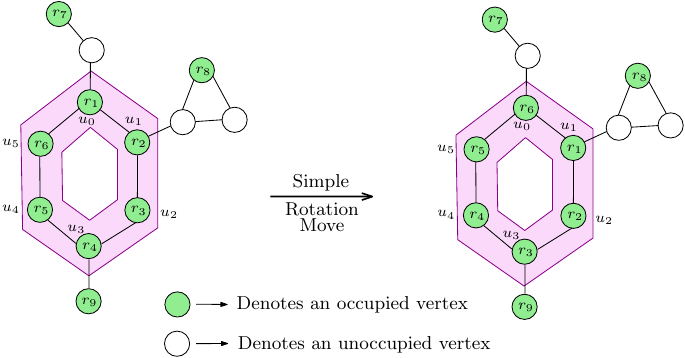}
    \caption{Simple rotation move, where $r_i$ denotes robot $i$.}\label{fig:simple_rot}
    \end{figure}

    \item \textbf{Dummy Move:}  A dummy move is the pair $(S,S)$ for any configuration $S$; this corresponds to not moving any robot from its occupied vertex in $S$.  
\end{itemize}

A \emph{valid move} is either a simple path move, a simple rotation move, or a dummy move. Note that we do not allow robots to swap along an edge as a move on its own. Trivially, $(S,S')$ is a valid move if and only if $(S',S)$ is a valid move as it is an undirected graph, and each of these valid moves are reversible.

Further, for a configuration $S$ and a permutation $\pi: [p] \rightarrow [p]$, their composition $S \circ \pi$, corresponds to another configuration where robots are renamed according to the permutation $\pi^{-1}$; if robot $i$ occupies vertex $S(i)$ in $S$, then robot $\pi^{-1}(i)$ also occupies vertex $S \circ \pi (\pi^{-1}(i)) = S(i)$ in the configuration $S \circ \pi$, as illustrated in \cref{fig:composition}. We state an observation regarding valid moves.

\begin{figure} [ht] 
    \centering
    \includegraphics[scale=0.8]{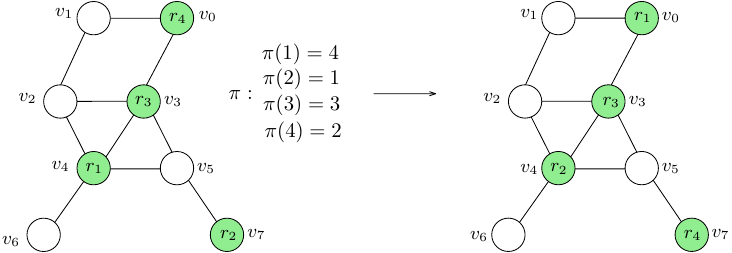}
    \caption{Composition of a configuration over a permutation, where $r_i$ denotes robot $i$.}\label{fig:composition}
\end{figure}

\begin{observation}\label{obs:valid_move_perm}
    For configurations $S$ and $S'$, $(S, S')$ is a valid move if and only if $(S \circ \pi, S' \circ \pi)$ is a valid move for any permutation $\pi: [p] \to [p]$.
\end{observation}

\begin{proof} 
Note that $S \circ \pi$ and $S' \circ \pi$ follow the same renaming of robots (according to $\pi^{-1}$). Hence, if $(S, S')$ is a valid move, so is $(S \circ \pi, S' \circ \pi)$, and vice versa.
\end{proof}

\subsection{Equivalence Relation based on Reachable Configurations}\label{subsec: eq-relation-reach}
We now define what is meant by a configuration to be reachable from another configuration and define a binary relation based on reachability. We will use this to define equivalence classes which would further help us in analysing the structure of the problem instance.

\begin{definition}[Reachability of configurations] We say that a configuration $T$ is reachable from $S$, if there is a finite sequence of configurations $(S_0, S_1, \ldots, S_t)$ with $S_0 = S$, $S_t = T$ where $(S_{k-1}, S_k)$ is a valid move for all $k \in [t]$. 
\end{definition}

Informally, $T$ is reachable from $S$ if, for all $i \in R$, robot $i$ starts at vertex $S(i)$, and after a sequence of valid moves, it reaches vertex $T(i)$. Since moves are reversible, for an undirected graph, if $T$ is reachable from $S$, then $S$ is also reachable from $T$. Further, $S$ is reachable from itself as a dummy move is a valid move.

Further, based on the above definition of reachability, a \emph{homogeneous binary relation} (a homogeneous binary relation on a set $X$ is a binary relation between $X$ and itself) $\Tilde{R}$ can be defined as follows:

\begin{definition}\label{def:tilde-R}
    $\Tilde{R}$ is the homogeneous binary relation over the set of all configurations $\mathcal S_G$ defined by reachability as,
    $$\Tilde{R} = \{(S, T) \mid S, T \in \mathcal S_G \text{ and }T\text{ is reachable from } S\}$$ 
    We denote $(S, T) \in \Tilde{R}$ as $S \sim T$.
\end{definition} 


The above relation $\Tilde{R}$ is an \emph{equivalence relation}, as 
\begin{itemize}
    \item $S \sim S$ for all $S \in \mathcal S_G$ (reflexive)
    \item $S \sim T \implies T \sim S$ for all $S, T \in \mathcal S_G$ (symmetric) 
    \item ($S \sim T$ and $T \sim U$) $\implies$ $S \sim U$ for all $S, T, U \in \mathcal S_G$ (transitive)
\end{itemize}

Since an equivalence relation provides a partition of the underlying set into disjoint equivalence classes, the reachability relation $\Tilde{R}$ also partitions the set $\mathcal S_G$ into disjoint equivalence classes based on mutual reachability. Further, we can generalise \cref{obs:valid_move_perm} to the notion of reachability as follows, by the same argument of renaming of the robots. 

\begin{observation}\label{obs:reachability_perm_composition}
    For configurations $S$ and $S'$, $S \sim S'$ if and only if $S \circ \pi \sim S' \circ \pi$ for any permutation $\pi : [p] \to [p]$.
\end{observation}

\begin{proof}
Suppose $S \sim S'$. Then there exists a sequence of configurations $(S_0 = S, S_1, \ldots, S_t = S')$ such that $(S_{k-1}, S_{k})$ is a valid move for all $k \in [t]$. From \cref{obs:valid_move_perm}, it follows 
 that the sequence $(S_0 \circ \pi = S \circ \pi, S_1 \circ \pi, \ldots, S_t \circ \pi = S' \circ \pi)$ is such that $(S_{k-1} \circ \pi, S_{k} \circ \pi)$ is a valid move for all $k \in [t]$. Thus, $S \circ \pi \sim S' \circ \pi$.
\end{proof}

\subsection{Problem Definition}\label{subsec: problem_defn}
We address the problem {\sc Universal Solvability of Robot Motion Planning on Graphs (USolR)}, which takes as input a graph $G(V, E)$ and a positive integer $p$ and asks if for $p$ robots all configurations $S$ are reachable from all configurations $T$.

\defproblem{\textsc{Universal Solvability of Robot Motion Planning on Graphs (USolR)}}{A graph $G(V,E)$, an integer $p \in \mathbb N$ ($2\leq p \leq |V|$)}{Decide if for all configurations $S$, $T$ of $p$ robots on the underlying graph $G$, are $S$ and $T$ reachable from each other using valid moves}

To solve this problem, we propose polynomial-time randomized and deterministic algorithms that require solving another related problem {\sc Feasibility of Robot Motion Planning on Graphs (FRMP)} as a subroutine. In the FRMP problem, we are given a graph $G(V, E)$, a positive integer $p$ denoting the number of robots, and a pair of robot configurations $(S, T)$, and we need to determine if the configuration $T$ is reachable from $S$ using valid moves.

\defproblem{\textsc{Feasibility of Robot Motion Planning on Graphs (FRMP)}}{A graph $G(V,E)$, an integer $p \in \mathbb N$ ($2 \leq p \leq |V|$), a pair of robot configurations $(S, T)$}{Decide if the configuration $T$ is reachable from $S$ using only valid moves on the underlying graph $G$}

 Yu and Rus~\cite{yu2014pebblemotiongraphsrotations} proposed a subroutine that solves the $\FRMP$ problem for a given instance consisting $(G(V, E), p, (S, T))$ in time $\mathcal O(|V| + |E|)$, based on the Theorem 20 in their paper.

We denote this algorithm as $\mathcal A$ in our work, and for the sake of completeness, we state Theorem 20 from their paper~\cite{yu2014pebblemotiongraphsrotations} again, as per the notations used in our paper.

\begin{proposition}\label{prop:feas-check-algo} 
    Given an instance $(G(V, E), p, (S, T))$ of $\FRMP$, there exists a deterministic algorithm, denoted as Algorithm~$\mathcal A$, to decide if $T$ is reachable from $S$ using only valid moves on $G$, that runs in time $\mathcal O(|V|+|E|)$.
\end{proposition}

\section{Accumulating Robots to Specific Vertices}\label{sec:accumulation}

In this section, we look into how to accumulate the robots into specific vertices using valid moves. It is important to note that we do not actually intend to use the accumulation algorithm as a subroutine, rather we prove structural properties of the output of such an algorithm on various inputs, and use the existence of such an algorithm to better understand the mathematical structure of our problem.

Assume that the underlying graph $G(V,E)$ is connected, with the set of vertices $V$ having some intrinsic total ordering ${\sf Intr}$, and let $R = [p]$ be the set of robots. Let \textsf{BFS} be the deterministic algorithm for breadth first search on $G(V,E)$ that enumerates neighbours in the order of the total ordering ${\sf Intr}$ of $V$, starting with the `least' indexed vertex as the root. Let $v_1, v_2, \ldots, v_n$ be the order in which \textsf{BFS} visits the vertices of $G$. Define $V_p = \{v_1, v_2, \ldots, v_p\}$. Notice that $V_p$ is fixed given $G(V,E)$ and the total ordering ${\sf Intr}$ of $V$. We now aim to reach a configuration $S'$ from $S$ using valid moves such that the set of vertices occupied in $S'$ is exactly $S'(R) = V_p$.
If $S(R) = V_p$, then the algorithm terminates and returns the configuration $S(R)$; hence we assume $S(R) \ne V_p$. Let $\mathfrak{T}$ be the BFS tree obtained by the algorithm \textsf{BFS}. Let $a \in [n]$ be the least index such that $v_a$ is unoccupied i.e., $v_a \notin S(R)$. Clearly $a \in [p]$. Let $b \in [n] \setminus [p] = \{p+1, p+2, \ldots, n\}$ be the least element such that $v_b$ is occupied by a robot, i.e., $v_b \in S(R)$. Let $P = (u_0 = v_b, u_1, u_2, \ldots, u_k = v_a)$ be the unique path between $v_a$ and $v_b$ respectively in $\mathfrak T$. Let $S'$ be the configuration obtained from $S$ by one simple path move along $P$. We now repeat this entire procedure on $S'$ till the resulting configuration does not occupy the vertex set $V_p$. The pseudocode of this algorithm, named Algorithm~\ref{alg:acc-alg} can be found in the Appendix.

Please refer to \cref{fig:accumulation} for a visual representation of the execution of the algorithm. By construction, if $S'$ is the output of \cref{alg:acc-alg} when the configuration $S$ is part of the input, then $S'(R) = V_p$. We now show that \cref{alg:acc-alg} terminates in finite time. 

\begin{figure} [ht] 
    \centering
    \includegraphics[scale=0.6]{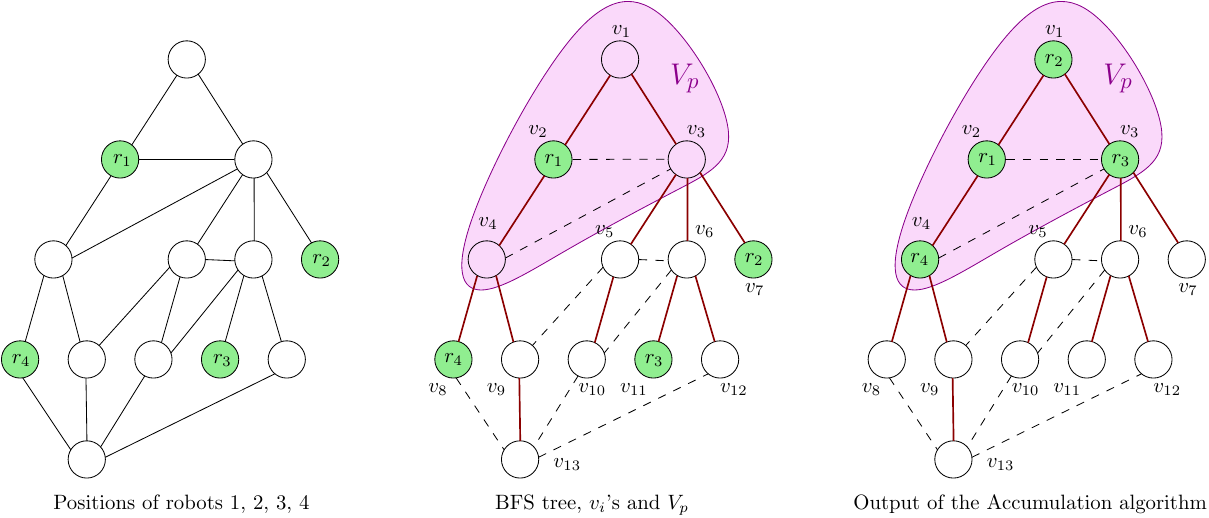}
    \caption{Illustration of working of \cref{alg:acc-alg}, where $r_i$ denotes robot $i$.}
    \label{fig:accumulation}
\end{figure}

\begin{lemma}
    \cref{alg:acc-alg} terminates in finite time.
\end{lemma}
\begin{proof}
    Let $d_v$ be the distance of a node $v$ to its nearest node in $V_p$ in the graph $\mathfrak{T}$. At each iteration, $d_{S(i)}$ can only decrease, for any $i \in R$. Moreover, there exists $i \in R$ such that $d_{S(i)}$ strictly decreases, because either some robot becomes included in $V_p$, or robots move along a path and are placed nearer to $V_p$. Therefore, the non-negative integer value $\sum \limits_{i \in R} d_{S(i)}$ strictly decreases at each iteration. Hence, \cref{alg:acc-alg} must terminate in finite time.
\end{proof}

We now argue that if the positions of the robots are permuted in $S$, the outputs of \cref{alg:acc-alg} are permuted correspondingly. 

\begin{lemma}[$\star$]\label{lem:acc-resp-perm}
    Consider a graph $G(V,E)$, a set of robots $R = [p]$ and an intrinsic ordering ${\sf Intr}$ of $V$. Let $S_1$ and $S_2$ be configurations such that they occupy the same set of vertices, i.e., $S_1(R) = S_2(R)$. Therefore, there exists a permutation $\pi: R \to R$ such that $S_1(i) = S_2(\pi(i))$ for all $i \in R$. Let $S'_1$ be the output configuration when $\cref{alg:acc-alg}$ is run on $G$ and $S_1$; and let $S'_2$ be the output configuration when \cref{alg:acc-alg} is run on $G$ and $S_2$. Then for all $i \in R$, $S'_1(i) = S'_2(\pi(i))$
\end{lemma}

Therefore, if robot $r$ occupies vertex $v$ in $S$ i.e., $S(r)$, then the final position of $r$ in $S'$ i.e., $S'(r)$ is only dependent on $v$ and the set $S(R)$; where $S'$ is the output of \cref{alg:acc-alg} on $G$ and $S$. This motivates us to define the accumulation map.

\begin{definition}[Accumulation map given a configuration $S$]
    Consider a graph $G(V,E)$, a set of robots $R = [p]$ and an intrinsic ordering ${\sf Intr}$ of $V$. For a configuration $S$, define the accumulation map given $S$, $\amap{S} : S(R) \to V_p$ as 
    $$\amap{S}(S(i)) = S'(i)$$
    where $S'$ is the output of \cref{alg:acc-alg} on $G$ and $S$.
\end{definition}

Thus, $\amap{S}$ maps a set of occupied vertices $S(R)$ to $V_p$. From \cref{lem:acc-resp-perm}, we have $\amap{S_1}$ and $\amap{S_2}$ are identical whenever $S_1(R) = S_2(R)$, i.e., the configurations $S_1$ and $S_2$ occupy the same set of vertices. Therefore, for a $p$-sized subset $\alpha \in \binom{V}{p}$, this allows us to uniquely define $\amapset{\alpha}: \alpha \to V_p$ to be equal to $\amap{S}$ for any configuration $S$ with $S(R) = \alpha$. We can hence define the output $Y_S$ of \cref{alg:acc-alg} when run on $G$ and $S$ as $Y_S = \amap{S} \circ S$, where $\circ$ is the composition function. Here, $S$ is a map from $R$ to $S(R) \subseteq V$, and $\amap{S}$ is a map from $S(R)$ to $V_p$, hence $Y_S = \amap{S} \circ S$ is a map from $R$ to $V_p$. 

\begin{definition}[Accumulation of a configuration]\label{def:Y_S}
    Consider a graph $G(V,E)$, a set of robots $R = [p]$ and an intrinsic ordering ${\sf Intr}$ of $V$. For any configuration $S$, we say the configuration $Y_S = \amap{S} \circ S$, the output of \cref{alg:acc-alg} on $S$, is the accumulation of $S$.
\end{definition}

Since \cref{alg:acc-alg} performs valid moves, we have the following observation.

\begin{observation}\label{obs:Y_s-reachable-from-S}
    Consider a graph $G(V,E)$, a set of robots $R = [p]$ and an intrinsic ordering ${\sf Intr}$ of $V$. For every configuration $S$, let $Y_S$ be the accumulation of $S$. Then, we have $Y_S \sim S$.    
\end{observation}

As $Y_S(R) = V_p \subseteq V$, we can view $Y_S$ as either a function from $R$ to $V_p$, or even $R$ to $V$, hence $Y_S(R)$ is also a configuration.

If $S$ is a configuration such that $S(R) = V_p$, the \cref{alg:acc-alg} stops without entering the loop at Line~\ref{line:acc-loop}. 
 Therefore, we get $Y_S = S$.

\begin{observation}\label{obs:invariant}
    Consider a graph $G(V,E)$, a set of robots $R = [p]$ and an intrinsic ordering ${\sf Intr}$ of $V$. For a configuration $S$ with $S(R) = V_p$, we have the accumulation $Y_s = S$, and the accumulation map $\amapset{V_p} : V_p \to V_p$ is the identity function.
\end{observation}

Recall that $(S,S')$ is a valid move if and only if $(S',S)$ is a valid move. 
Moreover, if $S$ and $T$ are configurations, then $S \sim Y_S$ and $T \sim Y_T$ (\cref{obs:Y_s-reachable-from-S}). This implies the following result.

\begin{observation}\label{obs:eqv-acc}
    Consider a graph $G(V,E)$, a set of robots $R = [p]$ and an intrinsic ordering ${\sf Intr}$ of $V$. Configurations $S$ and $T$ are reachable from each other if and only if their respective accumulations $Y_S$ and $Y_T$ are reachable from each other, i.e., $S \sim T \iff Y_S \sim Y_T$. 
\end{observation}

We now prove the following results regarding accumulation to any set of vertices. 

\begin{lemma}\label{lem:arbitrary-acc}
    Consider a graph $G(V,E)$, a set of robots $R = [p]$ and an intrinsic ordering ${\sf Intr}$ of $V$. Let $S: R \to V$ be any arbitrary configuration. Let $X \subseteq V$ be any arbitrary subset of size $|X| = p$. There exists a configuration $T: R \to V$ such that $T(R) = X$ and $S \sim T$.
\end{lemma}

\begin{proof}
    Consider the configuration $T: R \to V$ defined as $T = \mathcal{T}_{X}^{-1} \circ Y_S$. Since $\mathcal{T}_{X}^{-1} : V_p \to X$ is a bijection, and $Y_S: R \to V_p$ is a bijection, $T$ must also be a injective map with $T(R) = X$.

    Now, $Y_T = \mathcal{T}_{T(R)} \circ T = \mathcal{T}_{X} \circ (\mathcal{T}_{X}^{-1} \circ Y_S) = Y_S$. Since $S \sim Y_S$ and $T \sim Y_T$, this implies $S \sim T$. This completes the proof.
\end{proof}

\section{Randomized Algorithm for {\sc USolR}}\label{sec:randomized}

In this section, we design a polynomial-time randomized algorithm for the $\USOLR$ problem based on the sizes of the equivalence classes of the relation $\Tilde{R}$ (\cref{def:tilde-R}). The notions of the reachability relation and equivalence classes have resemblance to groups and subgroups. These similarities, as well as the group-theoretic techniques used in the works of Yu and Rus~\cite{yu2014pebblemotiongraphsrotations} motivated us to deeply examine the structure of these equivalence classes of the reachability relation $\Tilde{R}$, and we provide a characterisation of the equivalence classes of $\Tilde{R}$.

\subsection{Equivalence Classes of the Reachability Relation}\label{subsec: eq-class-reachability} 

Let $(G(V,E), p)$ be the instance of \USOLR, and let $V_p = \{v_1, v_2, \ldots, v_p\}$ be as defined in Section~\ref{sec:accumulation}. Consider the identity configuration $S_I$ with $S_I(i) = v_i$ for all $i \in [p]$. For configuration $S_I$, we have $S_I(R) = V_p$, so $Y_{S_I} = S_I$ (\cref{obs:invariant}). 

For a graph $G(V,E)$ to be universally solvable for $p$ robots, it must hold true that for any pair of configurations $S, T \in \mathcal S_G$, $T$ is reachable from $S$, i.e., $(S, T) \in \Tilde{R}$. Thus, for any configuration $S \in \mathcal S_G$, $S_I \sim S$ must also hold true. Further, from the discussion in \cref{sec:prelims}, we have that $\Tilde{R}$ is an equivalence relation, and since all configurations in $\mathcal S_G$ belong to the same equivalence class as $S_I$, there is essentially only one equivalence class comprising of all configurations $S \in \mathcal S_G$. Let us denote by $\Ex_{S}$, the set of configurations reachable from the configuration $S$. Then, for a YES-instance, we have $\Ex_{S_I} = \mathcal S_G$.

On the contrary, for a NO-instance, we have at least one configuration that is not reachable from the configuration $S_I$. Let $S_I$ be denoted by $S_0^*$.
Let $S_0^*, S_1^*, S_2^*, .. , S_k^*$ be configurations that are not reachable from each other, and the equivalence classes $\Ex_{S_0^*}, \Ex_{S_1^*}, \Ex_{S_2^*}, .., \Ex_{S_k^*}$ exhaust the entire set $\mathcal S_G$, i.e., $\mathcal S_G = \bigcup\limits_{i=0}^{k} \Ex_{S_i^*}$. Consider one of such configurations $S_j^*$, for $j \ge 1$ that is not reachable from $S_I (= S_0^*)$. 
We now show that the equivalence classes $\Ex_{S_I}$ and $\Ex_{S_j^*}$ have equal number of configurations. This essentially means $|\Ex_{S}| = |\Ex_{T}|$ for all configurations $S, T \in \mathcal S_G$. We will be crucially using this result to design an efficient randomized algorithm.

\begin{lemma}\label{lemma:eq-class-sizes} Consider a graph $G(V,E)$, a set of robots $R = [p]$ and an intrinsic ordering ${\sf Intr}$ of $V$. Let $\Tilde{R}$ be the equivalence relation on configurations as per~\cref{def:tilde-R}. Then the size of the equivalence class $\Ex_{S_I}$ equals the size of the equivalence class $\Ex_{S_j^*}$ for all $j \in \{1, 2, \ldots, k\}$.
\end{lemma}

\begin{proof}
    We have $S_I \not\sim S_j^*$. Let $Y_{S_j^*}$ denote the configuration obtained by running \cref{alg:acc-alg} on $S_1^*$ (\cref{def:Y_S}). By \cref{obs:Y_s-reachable-from-S}, $S_j^* \sim Y_{S_j^*}$. Therefore, $S_I \not\sim Y_{S_j^*}$ and $Y_{S_j^*} \in \Ex_{S_j^*}$. Since $Y_{S_j^*}$ is a bijection from $R$ to $V_p$, there exists a permutation $\pi^*$ such that $Y_{S_j^*} = S_I \circ \pi^*$ (i.e., $Y_{S_j^*}(i) = S_I(\pi^*(i)) = v_{\pi^*(i)}$ for all $i \in [p]$). Next, we define a 
    function $f_{\pi^*}: \Ex_{S_I} \to \Ex_{S_j^*}$ as $f_{\pi^*}(S) = S \circ \pi^*$, for all $S \in \Ex_{S_I}$. Note that the function $f_{\pi^*}$ is defined according to the permutation $\pi^*$ that relates $Y_{S_j^*}$ to $S_I$. We first establish that $f_{\pi^*}$ maps all configurations in $\Ex_{S_I}$ only to configurations in $\Ex_{S_j^*}$. With this, for $S_I$, we know that $f_{\pi^*}(S_I) = Y_{S_j^*} \in \Ex_{S_j^*}$. For any other configuration $\Tilde{S} \in \Ex_{S_I}$, there exists a sequence of configurations $S_1(=S_I), S_2, ...,  S_h(= \Tilde{S})$, where $(S_t, S_{t+1})$ is a valid move for every $t \in [h-1]$. Now, consider the sequence of configurations $f_{\pi^*}(S_1)= S_1 \circ \pi^* = Y_{S_j^*}, f_{\pi^*}(S_2) = S_2 \circ \pi^*, ...,  f_{\pi^*}(S_h) = \Tilde{S} \circ \pi^*$. In this sequence, $(f_{\pi^*}(S_t), f_{\pi^*}(S_{t+1}))$ is a valid move for every $t \in [h-1]$ (\cref{obs:valid_move_perm}). This corresponds to a sequence of valid moves from $Y_{S_j^*}$ to $f_{\pi^*}(\Tilde{S})$. Since the argument is independent of the configuration $\Tilde{S} \in \Ex_{S_I}$, we have that $f_{\pi^*}(\Tilde{S}) \in \Ex_{S_j^*}$,  for all $\Tilde{S} \in \Ex_{S_I}$. Thus, we have established that $f_{\pi^*}$ maps $\Tilde{S} \in \Ex_{S_I}$ to configurations only in $\Ex_{S_j^*}$. 

    Next, $f_{\pi^*}$ is an injective function as for any $S_1, S_2 \in \Ex_{S_I}$, $ (f_{\pi^*}(S_1) = f_{\pi^*}(S_2))$ implies $S_1 \circ \pi^* = S_2 \circ \pi^*$, and hence $S_1 = S_2$. Now, we argue that $f_{\pi^*}$ is also a surjective function. Consider a configuration $\Tilde{T}$ in $\Ex_{S_j^*}$, then the configuration $\Tilde{S} = \Tilde{T} \circ \pi^{*-1}$ is a valid configuration. Further, we observe that $\Tilde{S} \in \Ex_{S_I}$. Thus, $f_{\pi^*}(\Tilde{S}) = f_{\pi^*}(\Tilde{T} \circ \pi^{*-1}) = (\Tilde{T} \circ \pi^{*-1}) \circ \pi^* = \Tilde{T} \circ (\pi^{*-1} \circ \pi^*) = \Tilde{T}$ and hence $f_{\pi^*}(\Tilde{S}) \in \Ex_{S_j^*}$. Therefore, for every $\Tilde{T} \in \Ex_{S_j^*}$, there exists $\Tilde{S} = \Tilde{T} \circ \pi^{*-1} \in \Ex_{S_I}$ such that $f_{\pi^*}(\Tilde{S}) = \Tilde{T}$. Putting these together we obtain that $f_{\pi^*}$ is a bijective function, and hence $|\Ex_{S_I}| = |\Ex_{S_j^*}|$. This completes the proof. 
\end{proof}

This allows us to prove a crucial result, which helps us in designing an efficient randomized algorithm. 

\begin{lemma}\label{lem:solvability-eq-reln}
    Let $G(V, E)$ with $p$ robots be a NO-instance of $\USOLR$, and let $\mathcal S_G$ denote the set of all configurations over $G$, then there exist at least $|\mathcal S_G|/2$ many distinct configurations $S: R \to V$ in $\mathcal S_G$ such that $S$ is not reachable from $S_I$. 
\end{lemma} 

\begin{proof}
    From \cref{lemma:eq-class-sizes}, we have $| \Ex_{S_I} | = | \Ex_{S_i^*} |$ for all $i \in [k]$, and also $| \mathcal S_G| = \bigcup\limits_{i=0}^{k} | \Ex_{S_i^*} |$. Further, since the equivalence classes induce a partition over $\mathcal S_G$, we have $| \mathcal S_G| = \sum\limits_{i=0}^{k} | \Ex_{S_i^*} |$. Therefore, the size of each equivalence class is $| \Ex_{S_i^*} | = |\mathcal S_G|/(k+1)$ for all $i \in \{0, 1, .., k\}$. From this we can conclude that if for a given NO-instance $(G(V, E),p)$, the relation $\Tilde{R}$ is partitioned into $k+1$ equivalence classes $\Ex_{S_0^*}, \Ex_{S_1^*}, .., \Ex_{S_k^*}$ of equal sizes, then the configuration $S_I$ is not reachable to $(k \cdot |\mathcal S_G| )/(k+1)$ many configurations in $\mathcal S_G$. Since $k \geq 1$ for a NO-instance, $S_I$ is not reachable to $(k \cdot |\mathcal S_G| )/(k+1) \ge |\mathcal S_G|/2$ many configurations in $\mathcal S_G$.
\end{proof}

\subsection{Designing a randomized algorithm}\label{subsec:random-algo}

We now propose a randomized algorithm for $\USOLR$ based on the properties of equivalence classes discussed in \cref{subsec: eq-class-reachability}. Given a $\USOLR$ instance $(G(V, E),p)$, the randomized algorithm correctly outputs YES for a YES-instance (universally solvable instance) with probability $1$ and outputs YES for a NO-instance (not universally solvable instance) with probability at most $1/2$. 
We state the following theorem.

\begin{theorem}\label{thm:random-algo-analysis}
    Given a graph $G(V,E)$ and $p$ robots, there exists a randomized algorithm that correctly outputs YES for a universally solvable instance with probability 1, and outputs YES for a not universally solvable instance with probability at most $1/2$. The running time of this algorithm is $\mathcal O(|V| + |E|)$.
\end{theorem}

\begin{proof}

Given an instance $(G(V, E),p)$ for $\USOLR$, we propose a randomized algorithm, denoted by \cref{alg:random-alg}.
\cref{alg:random-alg} first finds the connected components of the given graph $G$ using a Depth-First Search algorithm. If there are $2$ or more connected components in $G$, then it outputs NO and terminates, as a disconnected graph is never universally solvable.
Next, \cref{alg:random-alg} samples a configuration $S_R$ uniformly at random from the set of all configurations $\mathcal S_G = \{S : R \to V \mid S \text{ is a configuration}\}$. Then, it runs Algorithm $\mathcal A$ (\cref{prop:feas-check-algo}) as a subroutine, with $(G(V, E),p,(S_I, S_R))$ as input to check if $S_R$ is reachable from $S_I$. If Algorithm $\mathcal A$ outputs YES for reachability of $S_R$ from $S_I$, then \cref{alg:random-alg} outputs that $G$ is universally solvable for $p$ robots (i.e., $(G(V, E),p)$ is a YES-instance), otherwise it outputs that $G$ is not universally solvable for $p$ robots (i.e., $(G(V, E),p)$ is a NO-instance). The pseudocode of~\cref{alg:random-alg} can be found in the Appendix.



    For a YES-instance $(G(V, E),p)$, all configurations $S \in \mathcal S_G$ are reachable from $S_I$. Therefore, when \cref{alg:random-alg} runs on a YES-instance, for any uniformly sampled configuration $S_R$, the feasibility subroutine Algorithm $\mathcal A$ will always return that $S_R$ is reachable from $S_I$, and hence \cref{alg:random-alg} would return YES with probability $1$.

    For a NO-instance, the number of distinct configurations that are reachable from $S_I$ is at most $(|\mathcal S_G|/2)$ (\cref{lem:solvability-eq-reln}). \cref{alg:random-alg} outputs YES for a NO-instance if Algorithm~$\mathcal A$ outputs that $S_R$ (the uniformly sampled configuration) is reachable from $S_I$. Therefore, the probability that a uniformly sampled configuration $S_R$ is reachable from $S_I$ is equal to the ratio of the number of configurations in $\mathcal S_G$ reachable from $S_I$ to the total number of configurations in $\mathcal S_G$.
    \begin{equation*}
        \Pr(\text{\cref{alg:random-alg} outputs NO for a YES instance}) \leq \frac{|\mathcal S_G|/2}{|\mathcal S_G|} \leq 1/2
    \end{equation*}

We now analyse the time complexity of \cref{alg:random-alg}. Sampling a random configuration $S_R$ from the set of all possible configurations can be done in linear time, $\OO(|V| + |E|)$. 
Algorithm $\mathcal A$ runs in linear time, $\OO(|V| + |E|)$ (\cref{prop:2-conn-solvability}). Moreover, all other steps (taking input, running DFS, etc.) are at most linear time, $\OO(|V| + |E|)$ operations. Hence, \cref{alg:random-alg} runs in time $\OO(|V| + |E|)$. This completes the proof of the theorem.
\end{proof}

\section{Deterministic Algorithm for {\sc USolR}}\label{sec:deterministic-algo}
In this section, we present a polynomial-time deterministic algorithm for the $\USOLR$ problem, which is a derandomization of \cref{alg:random-alg} presented in \cref{subsec:random-algo}. The algorithm uses as a subroutine, Algorithm $\mathcal A$ (discussed in \cref{prop:feas-check-algo}) to solve intermediate instances of $\FRMP$. 

For a given $\USOLR$ instance $(G(V, E)\text{, } p)$ to be universally solvable (or be a YES-instance), it must hold true that for any pair of configurations $(S, T), S, T \in S_G$, $T$ is reachable from $S$, i.e., $(S, T) \in \Tilde{R}$. Thus, by \cref{obs:eqv-acc}, $Y_S$ and $Y_T$ (as defined in \cref{def:Y_S}) must also be mutually reachable. Hence, it is sufficient to examine the mutual reachability of configurations of the form $Y_S$ (for some configuration $S$) mapping the set $R$ of $p$ robots to $V_p$.  There are $p!$ such configurations of the $p$ robots in $R$, each corresponding to a permutation $\pi \in \mathfrak{S}_p$ (the set of all permutation from $[p]$ to $[p]$), 
such that $S_\pi = S_I \circ \pi$ (i.e., $S_\pi(i) = v_{\pi(i)}$ for all $i \in [p]$). Thus, for a given instance to be universally solvable, all these configurations $\{S_\pi: R \to V_p \mid \pi \in \mathfrak{S}_p\}$ must be reachable from one another, or equivalently, all these configurations must be reachable from the identity configuration $S_I$, which corresponds to the identity permutation $\pi_I: [p] \to [p]$ defined as $ \pi_I(i) = i$, for all $i \in [p]$.

\begin{observation}\label{obs:uni-solvable}
    An instance $(G(V, E),p)$ of $\USOLR$ is universally solvable if and only if $S_I \sim S_\pi$ (i.e., $S_\pi$ is reachable from $S_I$) for all permutations $\pi : [p] \to [p]$.  
\end{observation}

A naive approach to decide universal solvability could be invoking Algorithm $\mathcal A$ for the $\FRMP$ instance $(G(V, E) \text{, } p \text{, } (S_I, S_\pi))$ corresponding to each of these $p!$ configurations (checking the reachability of each $S_\pi$ from the configuration $S_I$). However, this approach would have time complexity $\mathcal O(p! \cdot (|V| + |E|))$, which is super-polynomial in the size of the input parameter $p$. However, in the following lemma, we now argue that it is sufficient to check the reachability from $S_I$ of $p-1$ specific configurations from the set $\{S_\pi: R \to V_p \mid \pi \in \mathfrak{S}_p\}$. Consider the set of configurations $\{S_{\pi^*_t} \mid t \in [p-1] \}$, where the permutation $\pi^*_t$ is defined as follows.
\begin{equation*}
    \pi^*_t(i) = 
    \begin{cases}
      i & \text{if $i \neq t$ and $i \neq t+1$ } \\
      t + 1 & \text{if $i = t$}\\
      t & \text{if $i = t+1$}
    \end{cases}
\end{equation*}
Intuitively, the permutation $\pi^*_t$ is same as the identity permutation except that positions $t$ and $t+1$ are swapped.

\begin{lemma}\label{lem:naive-det-condition}
    Let $(G(V, E)\text{, } p)$ be an instance of $\USOLR$, then it is universally solvable if and only if $S_I \sim S_{\pi^*_t}$, for all $t \in [p-1]$ (i.e., each of the configurations $S_{\pi^*_t}$ is reachable from $S_I$) for the set of permutations $\{\pi^*_t \mid t \in [p-1]\}$ defined above.
\end{lemma}
\begin{proof}
    The forward direction of the proof is straightforward, because if $(G(V, E)\text{, } p)$ is a universally solvable instance, then by definition of the problem, all configurations $S \in \mathcal S_G$ are reachable from each other, therefore, the configurations $\{S_{\pi^*_t} \mid t \in [p-1]\}$ are also reachable from $S_I$.
    
    For the backward direction, we are given that $S_I \sim S_{\pi^*_t}$ for all $t \in [p-1]$. This implies that $S_\pi(= S_I \circ \pi) \sim S_{\pi^*_t} \circ \pi$, for all $t \in [p-1]$ for any permutation $\pi$ (\cref{obs:reachability_perm_composition}). This signifies that for any configuration $S_\pi$, using a sequence of valid moves, we can swap the robots occupying vertices $v_t$ and $v_{t+1}$ (while keeping other robots fixed), for every $t \in [p-1]$.
    To prove that an instance is universally solvable, we need to show that $S_I \sim S_\pi$ for all $\pi \in \mathfrak{S}_p$ (\cref{obs:uni-solvable}). Now, we show that $S_I \sim S_\pi$ for any $\pi \in \mathfrak{S}_p$ given $S_I \sim S_{\pi^*_t}$ for all $t \in [p-1]$. Consider any arbitrary permutation $\pi \in \mathfrak{S}_p$ and the corresponding configuration $S_\pi$. We incrementally reach the configuration $S_\pi$ from $S_I$ by routing the robots in the following manner: we first move the robot in $S_I$ that occupies the vertex $v_p$ in $S_\pi$. Let this robot be denoted as $j = \pi^{-1}(v_p)$. Then, in $S_I$, robot $j$ occupies $v_j$. Therefore, we first move robot $j$ from $v_j$ to $v_p$. 
    We argued that it is possible to swap robots in $v_t$ and $v_{t+1}$ for all $t \in [p-1]$ using a sequence of valid moves.
    We use this to establish that we can route the robot $j$ from $v_j$ to $v_p$. For that, initially, we have the configuration $S_0 = S_I$. The vertices $v_j, v_{j+1}, ... , v_p$ are occupied by robots $j, j+1, ..., p$ respectively. We can reach the configuration $S_{\pi^*_j}$ from $S_I$ (i.e., only the robots at vertices $v_j$ and $v_{j+1}$ are exchanged). 
    Now, we have robots arranged according to the configuration $S_{\pi^*_j}$. Further, we can exchange the robots at vertices $j+1$ and $j+2$ in this configuration as $S_{\pi^*_j}(= S_I \circ \pi^*_j) \sim S_{\pi^*_{j+1}} \circ \pi^*_j$ as discussed above. This can be repeated for $j+2, j+3, .., p-1$, so that now robot $j$ is at vertex $v_p$. In the current configuration $S_1$ (obtained by placing $1$ robot at its correct position in $S_\pi$), the robots $1, 2, .., j-1, j, j+1, .., p$ occupy vertices $v_1, v_2, .., v_{j-1}, v_p, v_j, .., v_{p-1}$ in the specific order. \\
    Now, we repeat the same procedure for the robot that occupies the vertex $v_{p-1}$ in $S_\pi$. Let this robot be denoted as $j' = \pi^{-1}(v_{p-1})$.  Then, in the current configuration $S_1$, robot $j'$ occupies the vertex $S_1(j')$. Now, we can again move the robot from $S_1(j')$ to $v_{p-1}$ following the same procedure as above. We can keep repeating this process, and reach configurations $S_2, S_3, .., S_p(=S_\pi)$ sequentially using valid moves. Hence, we can reach the configuration $S_\pi$ from $S_I$ for any arbitrary $\pi$. This shows that $S_I \sim S_\pi$, for all permutations $\pi :[p] \to [p]$. From \cref{obs:uni-solvable}, this implies $G$ is universally solvable for $p$ robots. This completes the proof of the backward direction, thereby proving the lemma.
\end{proof}

With this background, we introduce our deterministic algorithm, which we refer to as \cref{alg:deter-alg} and which derandomizes \cref{alg:random-alg}. We state the following theorem.

\begin{theorem}[$\star$]\label{thm:deter-alg}
    Given a graph $G=(V,E)$ and $p$ robots, there exists a deterministic algorithm that correctly outputs YES for a universally solvable instance and NO for an instance that is not universally solvable.
    The running time of this algorithm is $\mathcal O(p\cdot(|V| + |E|))$.
\end{theorem}

\section{Optimized Deterministic Algorithm for {\sc USolR}}\label{sec:optimised_deterministic}

In the previous section, we designed and analysed a polynomial-time deterministic algorithm solving \USOLR for $p$ robots and an input graph $G(V, E)$, with a running time guarantee of $\OO(p \cdot (|V| + |E|))$. Compared to the randomized algorithm, the running time blows up by a factor of $p$. In this section, we provide an optimised version of \cref{alg:deter-alg}, that exploits the structural properties of the input graph.

We start of by recalling Theorem 5 in the paper due to Yu and Rus~\cite{yu2014pebblemotiongraphsrotations}. We state this as the following proposition, as per the notation we introduced in this paper.

\begin{proposition}\label{prop:2-conn-solvability}
    Let $G$ be a 2-connected graph that is not a cycle and let there be $p$ robots. Then for all $p \le |V|$, $G$ is universally solvable for $p$ robots.
\end{proposition}

This allows us to prove the following result

\begin{lemma}\label{lem:2-conn-larger-than-p}
    For $p$ robots, a connected graph $G(V, E)$ with a 2-connected component of size at least $p$ that is not a cycle, is universally solvable.
\end{lemma}
\begin{proof}
    Let $C \subseteq V$ be a 2-connected component of size at least $p$ that is not a cycle. Let $C' \subset C$ be an arbitrary subset of $C$ having size $|C'| = p$.

    Let $S:R \to V$ and $T:R \to V$ be arbitrary configurations. Due to \cref{lem:arbitrary-acc}, there exists a configuration $S':R \to C'$ such that $S \sim S'$. Similarly, there exists a configuration $T':R \to C'$ such that $T \sim T'$. Further, due to \cref{prop:2-conn-solvability}, $S'$ and $T'$ are reachable from each other, just by using the vertices and edges in the graph induced by $C$; this gives $S' \sim T'$. Putting everything together, we have $S \sim S' \sim T' \sim T$.  This holds true for any pair of arbitrary configurations $S$ and $T$, hence $G$ is universally solvable for $p$ robots.
\end{proof}

With this, we are ready to state and prove the following lemma which is crucial to optimizing the deterministic algorithm.

\begin{lemma}[$\star$]\label{lem:bounded-edges}
    For $p \ge 2$, let $G(V,E)$ be a connected graph without any 2-connected component of size at least $p$ which is not a cycle. Then $|E| < p \cdot |V|$. 
\end{lemma}

This coupled with \cref{lem:2-conn-larger-than-p} gives us the following observation.

\begin{observation}\label{obs:no-instances-sparse}
    Every graph $G(V,E)$ that is not universally solvable for $p$ robots, must satisfy $|E| < p \cdot |V|$.
\end{observation}


With this, we are ready to state our final optimized algorithm.

\begin{theorem}\label{thm:usolr-complexity}
    Let $G(V, E)$ be a graph and let $p$ denote the number of robots. There exists a deterministic polynomial-time algorithm for the $\USOLR$ problem whose running time is bounded as follows:
    \[
    \begin{cases}
    \mathcal{O}(p \cdot (|V| + |E|)) & \text{if } |E| < p \cdot |V|, \\
    \mathcal{O}(|V| + |E|)  & \text{if } |E| \geq p \cdot |V|.
    \end{cases}
    \]
\end{theorem}
\begin{proof}
    Let us consider the instance $G(V,E)$, $p$ of \USOLR. If $G$ is disconnected, it is not universally solvable. Otherwise, by \cref{obs:no-instances-sparse}, whenever $|E| \ge p \cdot |V|$, the input instance must be a universally solvable. Therefore, we can check if $|E| \ge p \cdot |V|$ and output YES if the check passes. Otherwise we run \cref{alg:deter-alg} on the input instance to get the answer. 

    The correctness is of this algorithm is due to \cref{obs:no-instances-sparse} and the correctness of \cref{alg:deter-alg}. All that remains is to analyse the running time of this algorithm.

    If $G$ is disconnected, we detect in linear time, $\OO(|V| + |E|)$, and output NO correctly. If $|E| \ge p \cdot |V|$, we take input the graph, perform the check, and output YES correctly. This takes linear time, $\OO(|V| + |E|)$, in total. Otherwise, if $|E| < p \cdot |V|$, then we run \cref{alg:deter-alg}, which takes time $\OO(p \cdot (|V| + |E|))$ (\cref{thm:deter-alg}). This completes the proof.
\end{proof}

The pseudocode of this algorithm, called \cref{alg:opt-alg} can be found in the Appendix. \cref{alg:opt-alg} performs at least as good as \cref{alg:deter-alg} in all graphs. Moreover, when $p$ is small, and the graph is dense enough, \cref{alg:opt-alg} utilises \cref{obs:no-instances-sparse} to output YES correctly. Therefore, the speed up is significant for dense graphs with small number of robots.

\section{Augmenting Graphs for Universal Solvability}\label{sec:graph-augment} Having established the existence of polynomial-time deterministic and randomized algorithms for the $\USOLR$ problem, we now ask a natural follow-up question: Given an instance $(G (V, E), p)$ that is not universally solvable, can we augment the graph by adding a limited number of 
edges to make it universally solvable for $p$ robots? 
In this Section, we give combinatorial bounds on the minimum number of additional edges that need to be added to make the graph universally solvable for any number of robots $p$, so that the graph can reconfigure across any configurations of robots. 
We now formally state the problem below.

\defproblem{\textsc{Graph Edge Augmentation for Universal Solvability (EAUS)}}{A connected graph $G(V,E)$, the number of robots $p$ (such that $(G, p)$ is not universally solvable), budget $\beta$ for edge addition}{Decide if the graph $G$ can be made universally solvable for the given $p$ by adding at most $\beta$ new edges to $G$}

For the problem to be well-defined, we specify the number of robots $p$ for which we want $G$ to be universally solvable as part of the input. For the warehouse setting, where we would ideally want universal solvability for any number of robots $p \le |V|$, simply solving the above problem for $p=|V|$ guarantees universal solvability for all $p \le |V|$. This is because, for any configuration $S$ with the number of robots $p < |V|$, we can consider a configuration $S'$ of $|V|$ robots, with $p$ robots occupying the same vertices as in $S$ and $(|V|-p)$ dummy robots occupying the remaining $(|V|-p)$ empty vertices of $S$. Thus, if we can route between any two configurations of $|V|$ robots, we can route between any two configurations of a subset of size $p$ of $|V|$ robots, by transforming them into configurations of size $|V|$ by adding $(|V|-p)$ dummy robots.


We study the {\sc EAUS} problem only for connected graphs and further assume that the graph $G$ has at least four vertices. For the case where the number of vertices $|V|$  is less than or equal to three, a simplistic brute force approach can be used to solve {\sc EAUS}. On general graphs we do a case analysis for different classes of graphs based on vertex and edge connectivity. We rely upon certain structural results related to universal solvability discussed by Yu and Rus~\cite{yu2014pebblemotiongraphsrotations} for our analysis. We study combinatorial bounds on the number of edge additions $\beta$ for the problem for the different graph classes below. 

\subsection{2-connected graphs}
Here, we show that any $2$-connected graph is a YES-instance of {\sc EAUS} when $\beta = 1$. We analyse the case of 2-connected graphs in two subcases, as follows:

\subparagraph*{2-connected graphs that are not simple cycles:}
Previously, we discussed that the class of 2-connected graphs, that are not simple cycles, is universally solvable for all $p$ (\cref{prop:2-conn-solvability}). Therefore, these graphs do not require any edge additions for achieving universal solvability. The remaining case is that of simple cycles.

\subparagraph*{Simple cycles:}
For the case, when the graph $G(V, E)$ is a simple cycle, it is possible to make $G$ universally solvable for all $p \le |V|$ with just one extra edge addition. We can add the extra edge as any chord of the cycle $G$, thus forming a 2-connected component of size $|V|$, which is not a cycle. By \cref{prop:2-conn-solvability}, this becomes universally solvable. Therefore, an edge budget $\beta =1$ suffices for achieving universal solvability of simple cycles. 

\subsection{2-edge connected but not 2-connected graphs}


Any graph $G$ that is 2-edge connected but not 2-connected, is universally solvable for all $p < |V|$, as proved in Lemma 12 in Yu and Rus~\cite{yu2014pebblemotiongraphsrotations}. Thus, if such a graph is not universally solvable, it must be that $p =|V|$. We therefore discuss the {\sc EAUS} problem for this class of graphs for $p=|V|$ only. Since $G$ is 2-edge connected but not 2-connected, there are no cut edges in $G$, and there is at least one cut vertex in $G$. Now, before discussing this graph class, we first state a lemma that we need for our analysis.

\begin{lemma}\label{lem:odd-cactus}
    Given a graph $G(V,E)$ that is 2-edge connected but not 2-connected, it is not universally solvable for $p = |V|$ if and only if all the 2-connected components of $G$ are odd-length simple cycles. 
\end{lemma}
\begin{proof}
If $G(V,E)$ is a 2-edge connected but not 2-connected graph, then $G$ can be decomposed into 2-connected components sharing cut vertices between them. Now, each of these 2-connected components may or may not be a simple cycle. 

Firstly, we discuss the case when all the 2-connected components are simple cycles. Then, if all the 2-connected components are odd-length cycles, then the graph $G$ is not universally solvable for $p = |V|$, while even if one of the cycles is an even-length cycle, then the graph $G$ is universally solvable for $p = |V|$. This is because in the case of 2-edge connected but not 2-connected graphs, if all 2-connected components are odd-length cycles, all configurations are partitioned into two equivalence classes which are unreachable for each other, such that two configurations which are identical except any two adjacent vertices swapped, lie in different equivalence classes. However, if at least one of the cycles is of even length, it allows two configurations which are identical except two vertices swapped, to become reachable from each other, thus collapsing the two equivalence classes into one class. The detailed proof of this has been discussed in Theorem 2 of~\cite{yu2014pebblemotiongraphsrotations}.

Now, we consider the general case when at least one 2-connected component is not a simple cycle. Then, such a 2-connected component must contain at least 4 vertices. Now, consider one such 2-connected component $g$ of size at least 4 which is not a cycle, and denote its size as $s_{g}$. Then, we can route between any two configurations of $s_{g}$ robots in $g$ (\cref{prop:2-conn-solvability}). In particular, we can route between two configurations of $s_{g}$ robots, which are identical except that the robots on two adjacent vertices are swapped. Now, consider the configurations of $|V|$ robots over the entire graph $G$. Then, any two configurations of $|V|$ robots, which are identical except robots on two adjacent vertices of $g$ swapped, are also reachable from each other. Thus, the two equivalence classes of reachability (as there are in the case when all 2-connected components are odd-length cycles) again collapse into one class, implying universal solvability of $G$ for $p = |V|$.

Thus, a 2-edge connected but not 2-connected graph is universally solvable if and only if at least one of its 2-connected component is an even-length cycle or is a 2-connected component of size at least 4 which is not a cycle.
\end{proof}

Now, we first consider the special case where all the 2-connected components in $G$ are just triangles or $3$-cycles. For $p=|V|$, this graph class is not universally solvable (due to \cref{lem:odd-cactus}). To make it universally solvable for $p=|V|$, we consider two adjacent triangles $T_1$ and $T_2$ that share a common vertex, and add an edge joining one vertex of $T_1$ (one of the two vertices that are not shared with $T_2$) with one vertex of $T_2$ (again, one of the two vertices that not shared with $T_1$). This creates a 2-connected component of size 5 which is not a cycle. By \cref{lem:odd-cactus} this makes the graph universally solvable. 
Hence, adding just 1 extra edge (i.e., $\beta = 1$) suffices to make this graph class solvable.

 Thus, the other case that remains is when at least one of the 2-connected components of $G$ has size 4 or more. Now, for $G$ to be not universally solvable for $p = |V|$, all such 2-connected components of size 4 or more must be only simple odd-length cycles (due to \cref{lem:odd-cactus}). 
 We can add one extra edge to any one of these cycles to create a 2-connected component of size at least 4 that is not a simple cycle, making $G$ universally solvable for $p$ robots (again due to \cref{lem:odd-cactus}). 

 Thus, an edge budget $\beta = 1$ suffices to augment the class of 2-edge connected but not 2-connected graphs to achieve universal solvability for $p=|V|$.

\subsection{1-edge connected graphs}
For the class of 1-edge connected graphs, we provide both a lower bound and an upper bound (in terms of the number of robots $p$) for the number of edge additions required to make a graph $G$ universally solvable for $p$ robots. Thus, unlike the previous two classes of graphs, where a single edge addition suffices to achieve universal solvability, there are 1-edge connected graphs that do not become universal solvable by a constant number of edge additions.

\subparagraph*{Upper-bound.} 
We first prove an upper bound on the edge addition budget $\beta$ to achieve universal solvability for $p$ robots. We show that for a given instance $(G(V, E), p)$ which is not universally solvable, where $G$ is a 1-edge connected graph, at most $p-2$ edge additions suffice to make $G$ universally solvable for $p$ robots (i.e., $\beta \le p-2$). Clearly, since $G$ is not universally solvable for given $p$, $G$ must not have any 2-connected component which is not a cycle of size at least $p$ (\cref{lem:2-conn-larger-than-p}).
Now, consider any vertex $w$ in $G$ with degree at least 2. Let $x$ and $y$ be two of its neighbours. Consider $p-3$ other vertices and fix an arbitrary order of these $p-3$ vertices as $v_4, v_5, ..., v_p$. Now consider the vertex sequence $(x(=v_1), w(=v_2), y(=v_3), v_4, ..., v_p)$. $G$ already has $\{w,x\}$ and $\{w,y\}$ as edges.
Now, we add all edges connecting adjacent vertices $\{v_i, v_{i+1}\}$ (for convenience, define $v_{p+1} := v_1$) for all $i \in \{3, 4..., p \}$ to form a cycle. Some of these edges may already be present, so we add a maximum of $p-2$ edges this way. Now, we analyse the case of $p =|V|$ and $p < |V|$ separately.
\begin{itemize}
    \item For the case $p =|V|$, consider the structure of $G$ before we augmented $G$ to form the $p$-sized cycle. Since $|V| \ge 4$ and $G$ is connected, therefore, $G$ must have at least 3 edges (before addition of the extra edges). Two of these edges are $\{w,x\}$ and $\{w,y\}$. Therefore, there must have been at least one more edge in $G$ before the addition of extra edges. Now, if that edge connected two adjacent vertices $\{v_i, v_{i+1}\}$ for some $i \in \{3, 4, ..., p\}$ (again, define $v_{p+1} := v_1$), then we would have added at most $p-3$ edges to form the cycle. Then, we could add an extra edge between any two non-adjacent vertices of the cycle we created to form a chord, thus resulting in a 2-connected component of size $p$, which is not a cycle. Thus, we add a maximum of $p-2$ extra edges.
    If the extra edge was not between any two adjacent vertices of the cycle, then clearly, it must have been a chord between two non-adjacent vertices of the cycle (since all vertices are part of the cycle), and therefore, we again have a 2-connected component of size $p$ which is not a simple cycle. This makes the graph universally solvable (\cref{prop:2-conn-solvability}). Thus, we can achieve universal solvability for a general graph $G(V, E)$ and $p=|V|$ with a budget $\beta = p-2$ for edge additions in the worst case.
    \item For the case $p < |V|$, again consider the structure of $G$ before we augmented $G$ to form the $p$-cycle. 
    Since $G$ is assumed to be connected, there must be at least one edge in $G$ connecting a vertex that is part of the $p$-cycle to a vertex not in the cycle. Let $\{v_j, u\}$ be such an edge, where $v_j$ is part of the $p$-cycle formed by extra edge additions. We now state the following observation, which can be easily proven.

    \begin{observation}
        For $p$ robots, consider a connected graph $H$, which has a subgraph $H'$ of $p+1$ vertices with the following structure:
        \begin{itemize}
            \item $v_1, v_2, \ldots, v_{p-1}, v_{p}, v_{p+1} (=v_1)$ is a cycle.
            \item $v_0$ is a degree-1 vertex in $H'$ with $v_1$ as its neighbour. We call $v_0$ the `waiting vertex' in $H'$.
        \end{itemize}
        Then $H$ is universally solvable for $p$ robots.
    \end{observation}

     Therefore, the edge additions give us a $p$-cycle with a `waiting vertex' $u$ attached to $v_j$ appearing as a subgraph in $G$. Therefore, the entire augmented graph $G$ now becomes universally solvable for $p$ robots.
\end{itemize}
This proves the upper bound $\beta \le p-2$ to achieve universal solvability for $p$ robots for a 1-edge connected graph $G$.

\subparagraph*{Lower-bound.}
Now, we provide an asymptotically matching lower bound of $\Theta(p)$, even for the restricted class of 1-edge connected graphs. Consider the case $p=|V|$. Consider the star graph $G(V,E)$ with $|V|-1$ leaves. Note that for $p=|V|$, a robot can move only if it is located on a cycle. Since there are $|V|-1$ leaves, for all robots in these leaves to even be able to move, we need to at least make each leaf, a part of some cycle. This would require at least $(|V|-1)/2 = \Theta(|V|) = \Theta(p)$ additional edges. Moreover, for any tree with $\Theta(|V|) = \Theta(p)$ many leaves, at least $\Theta(|V|) = \Theta(p)$ additional edges would be required to make the graph universally solvable. We get a lower bound of $\beta \ge (|V|-1)/2$.

\begin{theorem}\label{thm:edge-aug-lb}
    There is an infinite family of graphs, such that for some graph $G(V,E)$ in it, for $p=|V|$ robots, universal solvability requires an edge budget of $\beta \ge \frac{|V|-1}{2} = \frac{p-1}{2}$.
\end{theorem}

Thus, we have both an upper-bound and a lower-bound for $\beta$ for 1-edge connected graphs, thus, showing that a constant budget for edge additions does not suffice to achieve universal solvability for such graphs.

The above graph classes: 2-connected, 2-edge connected but not 2-connected, and 1-edge connected exhaustively cover all graphs, thus making our combinatorial analysis complete.


\subsection{Augmenting a set of connected vertex and edges}

Addition of vertices along with addition of edges to an existing graph $G(V,E)$ in order to make it universally solvable is also an interesting variant of the augmentation problem. Looking back at the instances showing lower bound of $\beta$ in 1-connected graphs, where a $p-1$ star required $\beta = \Theta(p) = \Theta(|V|)$ additional edges to make it universally solvable, we observe that adding two more leaves to the central vertex of the star graph makes it universally solvable (it is easy to verify that a star graph with $p+1$ leaves is universally solvable for $p$ robots). Therefore, instead of a linear number of edge additions ($\beta = \Theta(|V|)$), addition of vertices along with edges allow us to make the graph universally solvable by addition of just $2$ edges and $2$ vertices. This motivates us to look into a variation of the augmentation problem that takes into consideration the addition of both vertices and edges.

In this variation, we consider the possibility to add $\alpha$ many vertices and $\beta$ many edges to the graph.

The $\beta$ edges can be added between the vertices of the input graph or these may connect a vertex from the input graph to one of the added $\alpha$ vertices or between the newly added $\alpha$ vertices.

We formally state this problem and then discuss cases of lower and upper bounds of $\alpha$ and $\beta$ over all graph classes.

\defproblem{\textsc{Graph Vertex and Edge Augmentation for Universal Solvability (VEAUS)}}{A connected graph $G(V,E)$, the number of robots $p$ (such that $(G, p)$ is not universally solvable), budget $\alpha$ for vertex addition, budget $\beta$ for edge addition}{Decide if the graph $G$ can be made universally solvable for the given $p$ by adding at most $\alpha$ vertices, and adding at most $\beta$ new edges to $G$ or between $G$ and the $\alpha$ vertices or between the newly added $\alpha$ vertices.}

\subparagraph*{Upper bounds.}Note that the upper bounds for {\sc Graph Edge Augmentation for Universal Solvability} hold true for {\sc Graph Vertex and Edge Augmentation for Universal Solvability} as well with $\alpha \ge 0$, since we can effectively keep these newly introduced vertices unused. We rewrite these for completeness:
\begin{itemize}
    \item For 2-connected graphs which are not cycles: The graph is already universally solvable.
    \item For cycles: $\beta \ge 1$ is sufficient for any $\alpha \ge 0$.
    \item For 2-edge connected but not 2-connected graphs: $\beta \ge 1$ is sufficient for any $\alpha \ge 0$.
    \item For 1-connected graphs: $\beta \ge p-2$ is sufficient for any $\alpha \ge 0$.
\end{itemize}

\subparagraph*{Lower bounds.} As discussed, the lower bound instance for {\sc Graph Edge Augmentation for Universal Solvability} allows solvability for $\alpha = 2$, $\beta = 2$. We show a stronger lower bound for a different set of graphs. 

\begin{definition}
Define the graph $Z_{\alpha, \beta}$ for some $\alpha$ and $\beta$ as :
\begin{itemize}
    \item $Z$ contains a vertex $v$, which we denote as the `center'.
    \item $(2\beta+1)$-many paths of length $(\alpha+1)$ are attached to the vertex.
\end{itemize}
\end{definition}

\begin{figure} [ht] 
    \centering
    \includegraphics[scale=0.65]{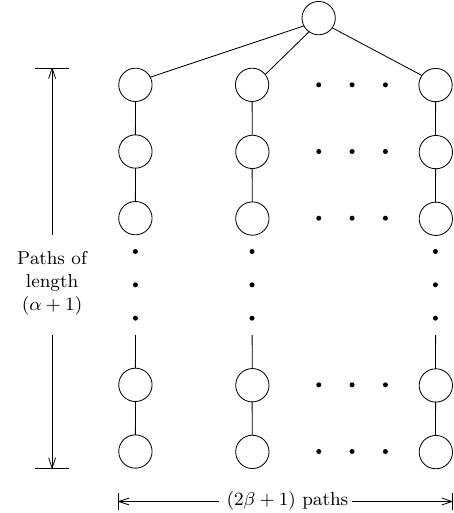}
    \caption{Schematic diagram of the graph $Z_{\alpha, \beta}$}\label{fig:beta_paths}
\end{figure}

We depict the graph $Z_{\alpha, \beta}$ in \cref{fig:beta_paths}. Therefore, the number of vertices $n$ of $Z_{\alpha,\beta}$ is equal to $1 + (2\beta + 1) \cdot (\alpha + 1)$. We will show that $Z_{\alpha, \beta}$ for $n$ robots is a NO-instance of {\sc Graph Vertex and Edge Augmentation for Universal Solvability} for vertex budget $\alpha$, and edge budget $\beta$. 

\begin{theorem}
    $Z_{\alpha, \beta}$ is a NO-instance of {\sc Graph Vertex and Edge Augmentation for Universal Solvability} for vertex budget $\alpha$, and edge budget $\beta$ and $1 + (2\beta + 1) \cdot (\alpha + 1)$ robots.
\end{theorem}
\begin{proof}
    Our proof exploits the fact that $Z_{\alpha, \beta}$ has large number of long paths. We denote by $n = 1 + (2\beta + 1) \cdot (\alpha + 1)$, the number of vertices of $Z_{\alpha, \beta}$

    Let the paths of length $\alpha + 1$ be named as $P_1, P_2, \ldots, P_{2\beta + 1}$. By pigeonhole principle, there must exist at least one such $P_i$ for which none among the $\beta$ newly added edges are incident on any vertex of $P_i$.

    We now look into the leave vertex along the path $P_i$, call this $u$. Consider a configuration $S$ with robot $1$ located on $u$, and consider another configuration $T$ with robot $1$ located on the center $v$. We will show that $T$ is not reachable from $S$.

    We define depth of a vertex as the distance from the center $v$. Note that the first time instance when robot $1$ is at a depth $h \in \{0,1,2,\ldots,\alpha+1\}$, the $P_i$ must have at least $(\alpha + 1 - h)$-many vacant vertices: all vertices between the current location of robot $1$ and the leaf vertex $u$. Therefore, the first time robot $1$ reaches the center (i.e., $h = 0$), there must be $\alpha + 1$ many empty vertices in $P_i$ (entire $P_i$ should be vacant). However this is a contradiction, because the augmented graph has at most $\alpha + n$ vertices and exactly $n$ robots, making it impossible to have more than $\alpha$ vacant vertices. This completes the proof.
\end{proof}

This gives us a lower bound of $\alpha = \Omega \left(\dfrac{|V|}{\beta}\right)$ over all graphs $G(V,E)$ to be a YES-instance for any number $p \le |V|$ of robots with vertex budget $\alpha$ and edge budget $\beta$. We note down a particularly interesting consequence of this, for $\beta = \Theta(\sqrt {|V|})$.

\begin{corollary}\label{cor:root-n-aug-lb}
    There is an infinite family of graphs, such that for some graph $G(V,E)$ in it, for $p=|V|$ robots, and edge budget $\beta = \Theta(\sqrt{|V|})$, universal solvability requires a vertex budget of $\alpha = \Omega(\sqrt{|V|})$.
\end{corollary}

\section{Conclusion}\label{sec:conclusion}
In this paper, we studied the problem of $\USOLR$ that asks whether a graph $G(V,E)$ is \emph{universally solvable} for $p$ robots. We derive that equivalence classes of mutually reachable configurations are of the same size. This leads to a randomized algorithm for the $\USOLR$ problem, running in time $\OO(|V| + |E|)$. We provide a derandomization of the above algorithm that runs in time $\OO(p \cdot(|V| + |E|))$, and further optimize it to run in time $\OO(p \cdot |E|)$ for sparse graphs with $|E| \le p \cdot |V|$, and in time $\OO(|E|)$ for other graphs. 

Lastly, we consider the {\sc EAUS} problem, where given a connected graph $G$ that is not universally solvable for $p$ robots, the question is to determine, if for a given $\beta$, at most $\beta$ edges can be added to $G$ to make it universally solvable for $p$ robots. Here, we analyse the problem based on the connectedness of the input graph. In particular, we provide an upper bound of $p-2$ as well as a lower bound of $\Theta (p)$ on $\beta$ for general graphs. Similarly, for the {\sc VEAUS} problem, we find examples of graphs that require large enough vertices and edges to be added in order to make the graph universally solvable.

An interesting problem for future work in the context of the two augmentation problems {\sc EAUS} and {\sc VEAUS} is to study these two problems on disconnected graphs.

\newpage
\bibliography{main}

\newpage
\appendix
\section{Missing details of Sections~\ref{sec:accumulation}-\ref{sec:optimised_deterministic}}
\subsection{Pseudocode of~\cref{alg:acc-alg}}

{\begin{algorithm}[H]
    \caption{\hfill \textbf{Input:} $G(V,E)$, a configuration $S$ of $p$ robots}\label{alg:acc-alg}
    \begin{algorithmic}[1]
        \State{$(v_1, \ldots, v_n) \gets$ BFS ordering from running \textsf{BFS} on $G$}
        \State{$\mathfrak T \gets$ BFS tree from running \textsf{BFS} on $G$}
        \State{$V_p \gets \{v_1, \ldots, v_p\}$}
        \While{$S(R) \ne V_p$}\label{line:acc-loop}
            \State{$a \gets \min\{i \in [p] \mid v_i \notin S(R)\}$}
            \State{$b \gets \min\{i \in [n] \setminus [p] \mid v_i \in S(R)\}$}
            \State{$P \gets $ unique path from $v_b$ to $v_a$ in $\mathfrak T$}
            \State{$S' \gets $ configuration obtained by pushing $S$ along $P$}
            \State{$S \gets S'$}
        \EndWhile
        \State{\Return $S'$}
    \end{algorithmic}
\end{algorithm}
}

\subsection{Proof of~\cref{lem:acc-resp-perm}}

\begin{proof}
    Notice that \cref{alg:acc-alg} is a deterministic algorithm and the execution is only dependent on the set of vacant and occupied vertices (and not the exact robots occupying the vertices). Therefore, at each iteration of \cref{alg:acc-alg}, we have $S(i) = S'(\pi(i))$. This is an invariant throughout the run of the \cref{alg:acc-alg} and this completes the proof.
\end{proof}

\subsection{Pseudocode of \cref{alg:random-alg}}

\begin{algorithm}[H]
    \caption{\hfill \textbf{Input:} $G(V,E)$,  $p = $ number of robots}\label{alg:random-alg}
    \begin{algorithmic}[1]
        \State{$n_{cc} \xleftarrow{}$ Number of connected components of $G$ (computed using DFS)}
        \If{$n_{cc} > 1$}
            \State{\Return NO}
        \EndIf
        \State{Define the Identity Configuration $S_I$ as $S_I(i) = v_i$ for all $i \in [p]$}
        \State{Sample uniform random configuration $S_R$ from the set of all configurations}
        \State{\Return output of Algorithm $\mathcal A$ on input $(G(V, E), p, (S_I, S_R))$} \Comment{\cref{prop:feas-check-algo}}
    \end{algorithmic}
\end{algorithm}

\subsection{Proof of \cref{thm:deter-alg}}
\begin{proof}
Given an instance $(G(V, E),p)$ for $\USOLR$, we state the deterministic algorithm, denoted as \cref{alg:deter-alg}, as follows. \cref{alg:deter-alg} takes as input an instance $(G(V, E)\text{, } p)$, and first finds the connected components of the underlying graph $G$ using a Depth-First Search. If there are $2$ or more connected components in $G$, then \cref{alg:deter-alg} returns that it is a NO-instance and terminates. \cref{alg:deter-alg} then checks the reachability of $S_{\pi^*_t}$ from $S_I$ by invoking as a subroutine, Algorithm $\mathcal A$ (\cref{prop:feas-check-algo}) for the $\FRMP$ instance $(G(V, E) \text{, } p \text{, } (S_I, S_{\pi^*_t}))$ for all $t \in [p-1]$. If $\mathcal A$ outputs a NO for any of these $p-1$ instances (i.e., if there exists $t \in [p-1] \text{ such that }S_{\pi^*_t}$ is not reachable from $S_I$), then \cref{alg:deter-alg} outputs that $G$ is not universally solvable for $p$ robots (i.e., $(G(V, E),p)$ is a NO-instance), else it outputs that $G$ is universally solvable for $p$ robots (i.e., $(G(V, E),p)$ is a YES-instance).

We formally describe \cref{alg:deter-alg} below.

\begin{algorithm}[H]
    \caption{\hfill \textbf{Input:} $G(V,E)$,  $p = $ number of robots}\label{alg:deter-alg}
    \begin{algorithmic}[1]
        \State{$n_{cc} \xleftarrow{}$ Number of connected components of $G$ (computed using DFS)}
        \If{$n_{cc} > 1$}
            \State{return NO}
        \Else
            \State{Define the Identity Configuration $S_I$ as $S_I(i) = v_i$ for all $i \in [p]$}
            \For{$t \gets 1 \text{ to } p-1$} \\
                \State{
                    Define $\pi^*_t$ such that $\pi^*_t(i) = 
                        \begin{cases}
                          i & \text{if $i \neq t$ and $i \neq t+1$ } \\
                          t + 1 & \text{if $i = t$}\\
                          t & \text{if $i = t+1$}
                        \end{cases}
                    $\\
                }
                \State{$temp \gets$ output of Algorithm $\mathcal A$ on input $(G(V, E), p, (S_I, S_{\pi^*_t}))$} \Comment{\cref{prop:feas-check-algo}}
                
                \If{$temp = \text{NO}$}
                    \State{return NO}
                \EndIf
            \EndFor
            \State{return YES}
        \EndIf
    \end{algorithmic}
\end{algorithm}

The correctness of \cref{alg:deter-alg} is immediate from \cref{lem:naive-det-condition}. 
We now analyse the time complexity of \cref{alg:deter-alg}. \cref{alg:deter-alg} takes as input a graph $G (V, E)$ and the number of robots $p$, and first performs a Depth-First Search to find the number of connected components of $G$, which takes at most linear time, $\mathcal O(|V| + |E|)$.  
Further, each of the $p-1$ reachability checks for $(S_I, S_{\pi^*_t})$ for $t \in [p-1]$ by invoking algorithm $\mathcal A$, run in time $\mathcal O(|V| + |E|)$ (\cref{prop:feas-check-algo}). Thus, the worst-case time complexity of \cref{alg:deter-alg} is $\mathcal O( (|V| + |E|) + (p-1)\cdot (|V| + |E|)) = \mathcal O(p\cdot(|V| + |E|))$. 

This completes the proof of the theorem.
\end{proof}

\subsection{Proof of~\cref{lem:bounded-edges}}
\begin{proof}
    Since the graph $G$ is connected, it can be decomposed into a tree of \emph{biconnected components} (maximal biconnected subgraphs) called the \emph{block-cut tree}~\cite{10.1145/362248.362272} of the graph. The biconnected components are attached to each other at shared vertices that are cut vertices.
    Let $B_g = \{a_1, a_2, ..., a_{b_1}, a_{b_1 + 1}, ..., a_b\}$ be the set of biconnected components of $G$ such that $a_1, a_2, ..., a_{b_1}$ are simple cycles, while $a_{b_1+1}, a_{b_1+2} ..., a_{b}$ are biconnected components that are not cycles. 
    
    Let $s_i$ denote the number of vertices in $a_i$. Given that $G$ does not have any biconnected component, which is not a cycle, of size at least $p$, we have,
    \begin{equation}\label{eq: 2-cc-sizes}
        s_i < p \text{ for all } i \in \{b_1+1, b_1+2, \ldots, b\}.
    \end{equation}


    Next, we bound the sum of the sizes of the biconnected components. For this, we construct a rooted tree $T$ whose vertices are the biconnected components as follows: choose a biconnected component of $G$ that shares only one cut-vertex with other biconnected components. Let us denote this component as $a^*$. It is easy to prove that there exists at least one such biconnected component that shares only one cut-vertex with other biconnected components, as otherwise the entire graph would be just one biconnected component. Now, we consider the distance of this biconnected component to itself to be 0. For any other biconnected component $a'$, we define its distance from $a^*$ to be the \emph{minimum} number of biconnected components (including $a^*$) that need to be traversed to go from $a^*$ to $a'$. For example, the biconnected components that share the single cut-vertex of $a^*$ are at a distance 1 from $a^*$ since we just need to traverse $a^*$ to reach them. Now, we construct the tree $T$ based on these distances, with the vertex corresponding to the component $a^*$ (with distance $0$) as the root vertex. Then, the vertices corresponding to the distance-1 components are connected as children of the root. Similarly, the distance-2 components are connected as children of the distance-1 component that is adjacent to them, in the shortest path to $a^*$. 
    
    Now, consider any non-leaf vertex $\nu$ of $T$. Let its number of children in $T$ be denoted by $d$, call these $\nu_1, \nu_2, \ldots, \nu_d$. The biconnected component associated with $\nu_i$, for any $i \in [d]$ shares exactly one cut vertex $u_i$ of $G$ with the biconnected component associated with $\nu$. For any other $j \ne i$, the biconnected component associated with $\nu_j$ shares some cut vertex $u_j$ of $G$ with the biconnected component associated with $\nu$; $u_j$ could be either same as $u_i$, or some other cut vertex in $G$. 
    
    Assume that the biconnected component associated with $\nu$ shares the same cut vertex with each of $\nu_1, \nu_2, \ldots, \nu_{d_1}$, for some $d_1 \le d$. Call this cut vertex $u$. Hence $d_1$ many edges, $\nu\nu_1, \nu\nu_2, \ldots, \nu\nu_{d_1}$, in $T$ are associated with the cut vertex $u$. When taking the sum of the sizes of all biconnected components, $\sum \limits_{i=1}^d s_i$, this cut vertex $u$ gets counted exactly $d_1+1$ times. This is same as the number of edges in $T$ that are associated with the cut vertex $u$, plus $1$. Note that this holds true for all cut vertices $u$ of $G$. For a vertex in $G$ that is not a cut vertex, the summation $\sum \limits_{i=1}^d s_i$ counts it exactly once.

    Let $q$ be the number of cut vertices in the graph $G$. The number of edges in $T$ is $b-1$. Since each edge of the tree $T$ corresponds to exactly one cut vertex in $G$, the total number of times all cut vertices in $G$ are counted in the summation $\sum \limits_{i=1}^d s_i$ is $q + (b-1)$. Moreover, since vertices in $G$ that are not cut vertices are counted exactly once in the summation $\sum \limits_{i=1}^d s_i$, these are counted $(|V| - q)$ times in total. Adding these counts we get,
    
    \begin{equation*}
        \sum_{i=1}^{b} s_i = |V| + b - 1.
    \end{equation*}
    
    Further, to obtain a bound on the number of biconnected components $b$, consider constructing the graph $G$ incrementally. We begin with a single vertex and iteratively add the remaining $|V|-1$ vertices along with their incident edges to incrementally create the biconnected components. Initially, with a single vertex, the count of biconnected components is zero. At each such step, adding a new vertex and connecting it to the existing graph can either expand an existing biconnected component or introduce at most one new biconnected component. Therefore, the total number of biconnected components is bounded above by the number of vertices minus one, i.e., $b \leq |V| - 1$. 
    Thus, we have,
    \begin{equation}\label{eq: cc-sizes-sum}
        \sum_{i=1}^{b} s_i = |V| + b - 1 \leq |V| + (|V| - 1) - 1 = 2|V| - 2.
    \end{equation}

    Now, we bound the number of edges $|E|$ in $G$. Each edge appears in exactly one biconnected component. Since the components $a_i$ for $i \in \{1, 2, ..., b_1\}$ are cycles, they have $s_i$ edges each. For the components  $a_i$ for $i \in \{b_1+1, b_1+2, ..., b\}$, each has at most $\binom{s_i}{2}$ edges. Thus,
    \begin{align*}
        |E| & \leq \sum_{i=1}^{b_1}s_i + \sum_{i=b_1+1}^{b}\binom{s_i}{2} & \\
        & < \frac{p}{2} \cdot \left(\sum_{i=1}^{b_1}s_i \right) + \sum_{i=b_1+1}^{b} \frac{s_i^2}{2} & \text{(as $p \ge 2$)} \\
        & < \frac{p}{2} \cdot \left(\sum_{i=1}^{b_1}s_i \right) + \frac{p}{2} \cdot \left(\sum_{i=b_1+1}^{b} s_i \right) & \text{(using \cref{eq: 2-cc-sizes})} \\
        & = \frac{p}{2} \cdot \left(\sum_{i=1}^{b}s_i \right) & \\
        & \le \frac{p}{2} \cdot (2|V|-2) & \text{(using \cref{eq: cc-sizes-sum})}\\
        & < p \cdot |V| &
    \end{align*}

    This completes the proof of the lemma.
\end{proof}

\subsection{Pseudocode of \cref{alg:opt-alg}}

\begin{algorithm}[H]
    \caption{\textsf{} \hfill \textbf{Input:} $G(V,E)$, $p$ = number of robots}
    \label{alg:opt-alg}
    \begin{algorithmic}[1]
        \State $n_{cc} \gets$ Number of connected components in $G$ (computed using DFS)
        \If{$n_{cc} > 1$}
            \State \Return NO
        \ElsIf{$|E| \ge p \cdot |V|$} \Comment{Must be a YES-instance, by \cref{obs:no-instances-sparse}}
            \State{\Return YES}
        \Else
            \State{\Return output of \cref{alg:deter-alg} run on $G$ and $p$}
        \EndIf
    \end{algorithmic}
\end{algorithm}

\end{document}